\documentclass[12pt, amsfonts, epsfig, reqno, letterpaper]{amsart}
\usepackage[hyperindex,breaklinks]{hyperref}
\usepackage{amsmath}
\usepackage{amsfonts}
\usepackage{mathrsfs}
\usepackage{graphicx}
\usepackage{color}
\usepackage{slashed}
\usepackage{braket}
\usepackage{extarrows}
\usepackage{amsmath, bm}
\usepackage{amssymb}
\usepackage[margin=1.5in]{geometry}
\usepackage{hyperref}

\newcommand{\Res}{{\mathrm{Res}}}

\newcommand{\C}{{\mathbb{C}}}

\newcommand{\V}{{{|V\rangle}}}
\newcommand{\Z}{{\mathbb{Z}}}

\newcommand{\bt}{{\mathbf t}}

\newcommand{\half}{\frac12}

\definecolor{yellow}{rgb}{1,.7,0}

\definecolor{pkured}{rgb}{0.55,0,0}

\newcommand{\be}{\begin{equation*}}
	\newcommand{\ee}{\end{equation*}}
\newcommand{\beq}{\begin{equation}}
	\newcommand{\eeq}{\end{equation}}

\numberwithin{equation}{section}

\newtheorem{cor}{Corollary}[section]

\newtheorem{lem}[cor]{Lemma}
\newtheorem{prop}[cor]{Proposition}
\newtheorem{thm}[cor]{Theorem}

\newtheorem{defn}[cor]{Definition}

\newtheorem{ex}[cor]{Example}
\theoremstyle{remark}
\newtheorem{rmk}[cor]{Remark}

\author{Shuai Guo}
\email{guoshuai@math.pku.edu.cn}
\address{Shuai Guo, School of Mathematical Sciences, Peking University, Beijing}

\author{Ce Ji}
\email{cji@tsinghua.edu.cn}
\address{Ce Ji,  Department of Mathematical Sciences, Tsinghua University, Beijing}

\author{Chenglang Yang}
\email{yangcl@pku.edu.cn}
\address{Chenglang Yang, Institute for Math and AI, Wuhan University, Wuhan 430072, China}
\address{Hua Loo-Keng Center for Mathematical Sciences,
	Academy of Mathematics and Systems Science,
	Chinese Academy of Sciences,
	Beijing}

\title[The bilinear fermionic form for KP and BKP hierarchies]{The bilinear fermionic form for KP and BKP hierarchies}

\begin{document}
	\maketitle
	
	\begin{abstract}
		For a tau-function of the KP or BKP hierarchy, we introduce the notion of lifting operator and derive an equation connecting the corresponding fermionic two-point function and fermionic one-point function through the lifting operator.
		This provides an effective approach to determine the fermionic two-point function of the tau-function from the lifting operator and the fermionic one-point function.
		As practical applications, we derive concise formulas for the  fermionic two-point functions of several models, like the $r$-spin model and the Br{\' e}zin--Gross--Witten model, which respectively  serve as examples for KP and BKP tau-functions.
	\end{abstract}

	\setcounter{section}{0}
	\setcounter{tocdepth}{1}
	\tableofcontents

	\section{Introduction}
	
	One of the central topics in modern enumerative geometry is to discuss how an enumerative theory is governed by certain integrable systems. The first attempt dates back to the celebrated Witten's conjecture initially proved by Kontsevich \cite{W90,K}, which states that the generating function of certain intersection numbers over the moduli spaces of stable curves (now known as the Kontsevich--Witten tau-function), is governed by the KdV hierarchy.
	A natural generalization is addressed to the $r$-spin theory. The tau-function of the $r$-spin model stores intersection numbers over the moduli spaces of $r$-spin curves and satisfies the higher KdV hierarchy \cite{W,FSZ}.
	Many other interplays between enumerative theories and integrable hierarchies were also studied, and until recently, most of the integrable systems can be regarded as certain reductions of KP hierarchy, including the Kontsevich--Witten and the $r$-spin case.
	Recently, relations between enumerative theories and the BKP hierarchy have been intensively studied \cite{MM,A23,LY2,LY3,LY1}.
	These developments have led us to study enumerative theories using methods based on these two integrable systems, KP and BKP hierarchies.
	
	Among the techniques for KP and BKP hierarchies, the fermionic approach is well established via the boson-fermion correspondence (see, for example, \cite{DJM}).
	A typical way introduced by Kyoto school in studying tau-functions of KP hierarchy (BKP hierarchy) is to assign it to a point in the Sato Grassmannian (isotropic Sato Grassmannian). For a tau-function $\tau(\mathbf{t})$ in the big cell, i.e. $\tau({\bf 0})\ne 0$, this can be characterized by the so-called canonical basis. 
	This canonical basis of a tau-function $\tau(\mathbf{t})$ can be packed into a generating series, called the fermionic two-point function of $\tau(\mathbf{t})$.
	In essence, all properties of $\tau(\mathbf{t})$ should be derived from its fermionic two-point function. Some discussions can be seen in \cite{TW, O02, Z15, DLM, DYZ, WY}.
	
	In this paper, our main goal is to study the fermionic two-point function for tau-functions of KP and BKP hierarchies (see Section~\ref{sec:main A} and \ref{sec:main B}). The key to our result is by introducing the {\it lifting operator}, which is a specially chosen Kac--Schwarz operator and generates the admissible basis of the (isotropic) Sato Grassmannian for the corresponding tau-function
	(see the Definitions \ref{def:lifting} and \ref{def:lifting B} for more details).
	Via the lifting operator of a tau-function, a relation between the fermionic two-point function and its (dual) fermionic one-point function is obtained.
	The main result concerning the KP hierarchy is the following
	\begin{thm}\label{thm:intro A}
		Suppose $l$ is a lifting operator for the tau-function $\tau_{KP}$ of the KP hierarchy,
		then
		\begin{equation}\label{int:keyequation}
			(l_u^*-l_v) \cdot \Psi(u,v)  = \Psi^*(u)\Psi(v),
		\end{equation}
		where $\Psi(u), \Psi^*(v)$ and $\Psi(u,v)$ are fermionic one-point function, dual fermionic one-point and fermionic two-point function of $\tau_{KP}$ respectively.
		See the Theorem \ref{thm:main KP} for the detailed version.
	\end{thm}

In many enumerative theories, $\Psi(u)$ and $\Psi^*(v)$ are widely studied and computed.
They are related to special values of wave and dual wave functions (see, for example, Section 6 in \cite{DJM}).
It will be interesting to obtain an explicit and compact formula for the fermionic two-point function $\Psi(u,v)$ since it records all the information of the corresponding tau-function.
Theorem \ref{thm:intro A} provides a powerful constraint to $\Psi(u,v)$, which generally determines $\Psi(u,v)$.
For example,
we can use the Theorem \ref{thm:intro A} to derive the fermionic two-point function of the $r$-spin model.
\begin{cor}\label{thm: intro r-spin}
	The fermionic two-point function for the $r$-spin theory is given by
	\begin{align}\label{eqn:main r-spin}
		\Psi(u,v)=\frac{\big(\sum_{a=0}^{r-1}(l_u^*)^a (l_v)^{r-1-a}\big) \cdot \Psi^*(u) \Psi(v)} {u^r-v^r},
	\end{align}
	where $l$ is the lifting operator for $r$-spin model given in equation \eqref{eqn:l r-spin} and $l^*$ is the adjoint operator of the lifting operator $l$.
	See the Corollary \ref{thm: r-spin} for more details.
\end{cor}

	Indeed, we obtain a compact formula for the fermionic two-point function of the generalized Kontsevich model
	and the $r$-spin model is its special case (see Section \ref{GKM}).
	When $r=2$, a formula for $\Psi(u,v)$ of the Kontsevich--Witten tau-function was obtained in \cite{TW,O01,O02} (see also \cite{Z13,BY,Z15,DYZ} for other methods),
	which is related to the famous Airy kernel.
	General $r$ case and generalized Kontsevich model were studied in literature.
	In \cite{ACM,A21},
	they obtained formulas for $\Psi(u,v)$ of these models as asymptotic expansions of some integrals.
	In \cite{DLM}, they derived a recursive formula for coefficients of expansion of $\Psi(u,v)$ and applied it to compute several leading terms.
	Corollary \ref{thm: intro r-spin} gives a simple and closed formula for $\Psi(u,v)$.
	
	\medskip
	
	For the BKP hierarchy case, we also introduce the notion of the lifting operator of B-type and then establish an analog of Theorem \ref{thm:intro A} for B-type fermionic two-point function.
	Application to the Br{\' e}zin--Gross--Witten model is also obtained.
	More details can be seen in Sections \ref{sec:main B} and \ref{sec:BGW}.
	Our main result is
	\begin{thm}
		\label{thm:intro B main}
		Suppose $l$ is a lifting operator for the tau-function $\tau_{BKP}$ of the BKP hierarchy,
		and $\tilde{l}$ is its anti-symmetrization.
		Then we have
		\begin{align}
			(\tilde{l}_u+\tilde{l}_v)\cdot\Psi^B(u,v)
			=\Psi^B(u)\tilde{\Psi}^B(v)
			-\Psi^B(v)\tilde{\Psi}^B(u),
		\end{align}
		where $\Psi^B(u), \tilde{\Psi}^B(v)$ are the first and second fermionic one-point functions, and $\Psi^B(u,v)$ is the fermionic two-point function of $\tau_{BKP}$ respectively.
		See the Theorem \ref{thm:B main} for the detailed version.
	\end{thm}
	
In the last part of this paper (Section~\ref{futherexamples}), we gather a number of well-known and typical examples. Then, we apply our main results (Theorems \ref{thm:intro A} and \ref{thm:intro B main}) to these examples. Amazingly, we find that our formula can always function effectively: It computes the two-point function simply by using the one-point function and the lifting operator. This provides more evidence that the lifting operator plays a crucial role in this entire framework.

	\medskip
	
	The rest of this paper is organized as follows.
	In Section \ref{fermions}, we review the charged fermions and neutral fermions,
	together with the KP and BKP hierarchies.
	In Section \ref{walgebra}, we review the constructions of $W_{1+\infty}$ algebra and its BKP analogue $W^B_{1+\infty}$.
	We introduce the lifting operators and prove our main results, Theorems \ref{thm:intro A} and \ref{thm:intro B main} in Section \ref{mainresult}.
	Finally in Section \ref{examples},
	we compute the fermionic two-point functions of the $r$-spin model and the BGW model, which serve as examples of KP and BKP hierarchies respectively. We also list various kinds of other models where our main theorem works perfectly.

\vspace{.2in}
{\em Acknowledgements}.
The first author is partially supported by the National Key R \& D Program of China (No. 2023YFA1009802) and NSFC (No. 12225101).
The second author is partially supported by Tsinghua University (No. 2024SM349).
The third author is partially supported by the NSFC (No.  12288201, 12401079),
the China Postdoctoral Science Foundation (No. 2023M743717)
and China National Postdoctoral Program for Innovative Talents (No. BX20240407).
\vspace{.2in}

	\section{Fermions and neutral fermions}\label{fermions}
	In this section, we review the KP and BKP hierarchy, as well as their boson-fermion correspondences.
	
	\subsection{Free Fermions}
	
	The charged fermions are operators $\{\psi_k,\psi^*_{k}\}_{k\in\mathbb{Z}+\frac{1}{2}}$
	satisfying the following anti-commutation relations:
	\begin{equation}\label{eqn:fercomm}
		[\psi_i,\psi_j^*]_+ = \delta_{i,j}\cdot\mathrm{id},\ \ \ \
		[\psi_i,\psi_j]_+ = [\psi^*_i,\psi_j^*]_+=0,
	\end{equation}
	where the anti-commutator is $[a,b]_+:=ab+ba$.
	Together with the identity operator, fermions form an infinite dimensional Lie algebra.
	A kind of realization of this Lie algebra is based on the semi-infinite wedge product.
	
	We start with an infinite dimensional complex vector space $V=z^{1/2}\cdot\C[z,z^{-1}]]$, with a linear basis $\{\underline{k+\frac{1}{2}}=z^{k+\frac{1}{2}}\}_{k\in \mathbb Z}$ denoted by half integers.
	The fermionic Fock space $\mathcal{F}=\Lambda^{\frac{\infty}{2}}V$ is then defined as the semi-infinite wedge space of $V$.
	To be precise, let $S=\{s_1>s_2>\ldots\}$ be an admissible ordered infinite set of half-integers in $\mathbb{Z}+\half$,
	which means the following two sets are both finite,
	$$
	|S_+| = | \{s_i\in S, s_i>0\}|<\infty,\quad  |S_-|=  | \{s_i\in \mathbb{Z}_{\leq 0}-\tfrac{1}{2}, s_i\notin S\}|<\infty.
	$$
	Then the fermionic Fock space $\mathcal{F}$ is spanned by vectors characterized by $S$,
	$$
	\mathcal{F}=\mathrm{span}\left\{\sum_{S\  \mathrm{admissible}} c_{S}\cdot |S\rangle\right\}, \quad c_{S}\in \mathbb{C},
	$$
	where
	$$
	|S\rangle = \underline{s_1}\wedge \underline{s_2}\wedge\underline{s_3}\wedge\cdots.
	$$
	The fermionic Fock space admits a decomposition with respect to the charge of vectors. For $|S\rangle\in \mathcal{F}$, its charge is defined to be the integer
	$$
	\mathrm{charge}(|S\rangle):=|S_+|-|S_-|.
	$$
	Then the decomposition of $\mathcal{F}$ is
	$$\mathcal{F}=\bigoplus_{n\in \mathbb{Z}}\mathcal{F}^{(n)},$$
	where $\mathcal{F}^{(n)}$ is spanned by vectors $|S\rangle$ of charge $n$.
	The vacuum state of charge 0 is
	$$|0\rangle:=\underline{-\half}\wedge\underline{-\frac{3}{2}}\wedge\underline{-\frac{5}{2}}\wedge\cdots.$$
	
	The action of fermions $\psi_k$ on a basis of $\mathcal{F}$ are defined by
	$$
	\psi_k\cdot  |S\rangle := \underline{k} \wedge  |S\rangle,\quad k\in\Z+\half,
	$$
	and the operators $\psi^*_k$ are adjoint operators of $\psi_k$,
	whose actions are defined by
	\begin{equation*}
		\psi_k^*\cdot  |S\rangle:=\left\{
		\begin{split}
			&(-1)^{l+1}\underline{s_1}\wedge \underline{s_2}\wedge\cdots\wedge\widehat{\underline{s_l}}\wedge\cdots, \quad &&\mathrm{if}\ s_l=k\ \mathrm{for\ some\ }l;  \\
			&0, \quad &&\mathrm{otherwise}.
		\end{split}\right.
	\end{equation*}
	Since we have
	\[\psi_{-k} |0\rangle = \psi^*_{k}|0\rangle
	=0,\ \ k\in\mathbb{Z}_{+}+\frac{1}{2},
	\]
	the operators $\{\psi_r,\psi^*_{-r}\}_{r>0}$ are called the creators,
	and $\{\psi_r,\psi^*_{-r}\}_{r<0}$ are called the annihilators.
	One can notice that the Fock space $\mathcal{F}$ can be generated by the actions of creators on the vacuum state $|0\rangle$.

	There is a standard inner product over the fermionic Fock space $\mathcal{F}$.
	An equivalent formalism is given by the following.
	Let $\varphi_i$ be a linear summation of fermions. For any expression
	\begin{align*}
		\langle0|\varphi_1\cdots\varphi_n|0\rangle,
	\end{align*}
	one assigns a complex number called the vacuum expectation value.
	The vacuum expectation values are determined by the following conditions:
	linearity for each $\varphi_i$,
	the initial value $\langle0|0\rangle=0$,
	the action of left annihilators $\langle0|\psi_k=\langle0|\psi^*_{-k}=0,\ k\in\mathbb{Z}_++\half$,
	and the anti-commutation relations \eqref{eqn:fercomm}.
	An explicit construction of the vacuum expectation value and inner product can be seen in \cite{DJM}.
	For computation, we only need
	\begin{align}\label{eqn:<2psi> KP}
		\begin{split}
		\langle0|\psi^*_k\psi_j|0\rangle
		=&\langle0|\psi_{-k}\psi^*_{-j}|0\rangle=\delta_{k,j}\delta_{j>0}.,\\
		\langle0|\psi\psi|0\rangle
		=&\langle0|\psi^*\psi^*|0\rangle=0
		\end{split}
	\end{align}
	and in general,
	there is the Wick theorem (see Theorem 4.1 in \cite{DJM}), which decomposes general vacuum expectation values into the above simple cases \eqref{eqn:<2psi> KP}.
	
	The fermionic normal-ordering of two fermions $\varphi_1,  \varphi_2$ is defined as
	\begin{equation}\label{def_normalorder}
		:\varphi_1\varphi_2: =
		\varphi_1 \varphi_2
		-\langle0|   \varphi_1 \varphi_2 |0\rangle,
	\end{equation}
	where $\varphi_1, \varphi_2$ can represents $\psi_i$ or $\psi^*_i$.
	It is also convenient to define the fermionic fields, which are generating series of the fermionic operators
	\begin{equation}\label{defferfield}
		\psi(z) := \sum_{k\in\mathbb{Z}+\frac{1}{2}}z^{k-\frac{1}{2}}\psi_k,\quad \psi^*(z):=\sum_{k\in\mathbb{Z}+\frac{1}{2}}z^{-k-\frac{1}{2}}\psi^*_k.
	\end{equation}
 An easy calculation employing equations \eqref{eqn:fercomm} and \eqref{def_normalorder} shows the following operator product expansion (OPE)
\begin{equation}\label{ferppd}
	\psi^*(u)\psi(v)= \frac1{u-v}  + :\psi^*(u)\psi(v): .
\end{equation} where $\frac{1}{u-v}$ is expanded in the region $|u|>|v|$ (see, for example, \cite{DJM}).
	
	\subsection{Boson-fermionic correspondence}
	
	We consider
	\begin{align}\label{eqn:J_n}
		J_n:=-\sum_{k\in\mathbb{Z}+\frac{1}{2}} :\psi^*_{k}\psi_{k-n}:,
		n\in\mathbb{Z},
	\end{align}
and introduce the following vertex operator
	\begin{align*}
		\Gamma_-(\mathbf{t})=\exp\Big(\sum_{n\in\mathbb{Z}_{+}}
		t_n J_{-n}\Big),
	\end{align*}
	where $\mathbf{t}=(t_1,t_2,...)$ are called the time variables of KP flow.
	
	The boson-fermionic correspondence tells that the fermionic Fock space and bosonic Fock space are two equivalent representations of the Heisenberg algebra (see \cite{DJM}).
	Thus,
	it gives an isomorphism between the ring of formal power series with variables $\mathbf{t}=(t_1,t_2,...)$ and the fermionic Fock space as vector spaces.
	For convenience, we restrict us to the charge 0 component.
	Other components are isomorphic to this part.
	\begin{prop}[boson-fermionic correspondence]
		The following map
		\begin{align*}
			\Phi:\quad \mathcal{F}^{(0)} \quad&\rightarrow \quad\mathbb{C}[\![\mathbf{t}]\!]\\
			|V\rangle \quad&\mapsto \quad\langle0 |\Gamma_-(\mathbf{t}) |V\rangle
		\end{align*}
		is an isomorphism between vector spaces.
	\end{prop}
	
	\subsection{KP hierarchy}
	We follow Sato's theory \cite{S} to review the notion of tau-function of the KP hierarchy.
	For convenience, we restrict us to the charge 0 component $\mathcal{F}^{(0)}$.
	Since the infinite wedge formalism of the fermionic Fock space $\mathcal{F}^{(0)}=\big(\Lambda^{\frac{\infty}{2}}V\big)^{(0)}$,
	there is a natural Pl{\" u}cker embedding from the infinite dimensional Grassmannian $Gr$ to the projective space $\mathbb{P}\mathcal{F}^{(0)}$ since the construction using the infinite wedge (see \cite{DJM}).
	A general element in $Gr$ can be represented as
	\[f_0(z)\wedge f_1(z) \wedge\dots,\]
	where each $f_i(z)\in V=z^{1/2}\cdot\mathbb{C}[\![z^{-1}]\!][z]$.
	
	An element $|V\rangle\in\mathcal{F}^{(0)}\setminus\{0\}$ is a fermionic tau-function of the KP hierarchy means that its image of projection to $\mathbb{P}\mathcal{F}^{(0)}$ further lies in the Grassmannian $Gr$.
	After boson-fermionic correspondence,
	the conditions of being in the Grassmannian are equivalent to the Hirota bilinear relations of the KP hierarchy.
	For convenience, we use the same notation $|V\rangle$ to represent its projection to Sato's Grassmannian since a non-zero constant will not influence the KP integrability.
	When the projection of $|V\rangle$ is in the big cell of Sato's Grassmannian,
	writing as $|V\rangle \in Gr_{(0)}$,
	it can be represented as
	\begin{align*}
		|V\rangle=c\cdot\varphi_0\wedge\varphi_1\wedge\dots,
	\end{align*}
	where $c$ is a non-zero constant, and the so-called admissible basis $\varphi_k$ has the following form
	\begin{equation}\label{adba}
		\varphi_k = z^{k+\frac{1}{2}}+\sum_{i=1-k}^{\infty}a_{k,i}z^{-i+\frac{1}{2}},
		\quad a_{k,i}\in\mathbb{C},
		\quad k=0,1,2,\cdots.
	\end{equation}
	It is worth mentioning that, for any given $|V\rangle$, the admissible basis is not unique.
	By Gauss elimination, a kind of canonical choice can be applied.
	That is to say, for each $|V\rangle\in Gr_{(0)}$,
	there is a unique admissible basis of $|V\rangle$, which has the following form
	\begin{equation}\label{canobasis}
		\tilde\varphi_k = z^{k+\frac{1}{2}}+\sum_{i=0}^{\infty}b_{k,i}z^{-i-\frac{1}{2}}, \quad k=0,1,2,\cdots.
	\end{equation}
	They are called the canonical basis of $|V\rangle$ and can be obtained by Gauss elimination from any given admissible basis.
	Those unique coefficients $b_{k,i}$  appear in \cite{ADKMV} as the coefficients in the Bogoliubov transform of the fermionic vacuum.
	The big cell of the Sato's Grassmannian is an affine space,
	and then these coefficients $b_{k,i}$ are indeed the coordinates of the point corresponding to $|V\rangle$ in this affine space
	(see, for example, \cite{DJM, Z13, BY, Z15}).
	They can be obtained from the following formula
	\begin{equation}\label{eqn:bki}
		b_{k,i} = \frac{\braket{0| \psi^*_{k+\half} \psi_{-i-\half} |V}}{\braket{0|V}}.
	\end{equation}
	Conversely, the vector $|V\rangle$ can be recovered via its canonical basis (see \cite{DJM,ADKMV,Z15})
	\begin{align}\label{eqn:V A}
		|V\rangle=c\cdot e^A |0\rangle,
	\end{align}
	where
	\begin{equation}\label{eqn:A affine_coord}
		A= \sum_{k,i\geq 0}  b_{k,i}   \psi^*_{-i-\half} \psi_{k+\half}.
	\end{equation}
	Thus, one can notice that, up to the nonzero constant $c$, all information of the vector $|V\rangle$ is encoded in its canonical basis.
	As a consequence, studying the generating series of the canonical basis is an effective method to study the tau-function of the KP hierarchy (see \cite{O02,Z15}).
	Indeed, the main object in this paper, the fermionic two-point function, is a natural generating series of the canonical basis.

	\subsection{Neutral fermions}
	The neutral fermions are operators $\{\phi_i\}_{i\in\mathbb{Z}}$ satisfying
	\begin{align}\label{eqn:B comm. relations}
		[\phi_i,\phi_j]_+=(-1)^i\delta_{i+j,0}\cdot\mathrm{id}.
	\end{align}
	The B-type fermionic Fock space is generated by the vacuum state $|0\rangle$, which satisfies
	\begin{align}\label{eqn:B right|0>}
		\phi_i |0\rangle=0,
		\ \ \ \ \forall i<0.
	\end{align}
	That is to say,
	\begin{align*}
		\mathcal{F}_B
		=\Big\{\sum c_{k_1,...,k_l}\phi_{k_1}\cdots\phi_{k_l}
		|0\rangle|l\geq0, k_1>\cdots>k_l\geq0, c_{k_1,...,k_l}\in\mathbb{C}\Big\}.
	\end{align*}
	When asking the numbers of neutral fermions $l$ to be even, we obtain the subspace $\mathcal{F}_B^{(0)}$.
	Similarly, one can construct the dual B-type fermionic Fock space $\mathcal{F}_B^*$.
	It is a $\mathbb{C}$-vector space generated by
	\begin{align*}
		\langle0|\phi_{n_l}\cdots\phi_{n_2}\phi_{n_1},\ \ n\geq0, n_1<n_2\dots<n_l\leq0,
	\end{align*}
	where $\langle0|$ is the dual vacuum state satisfying
	\begin{align}\label{eqn:B left<0|}
		\langle0|\phi_{i}=0,\ \ \forall i>0.
	\end{align}
	Then there is a pairing $\mathcal{F}_B^*\times\mathcal{F}_B\rightarrow\mathbb{C}$ determined by equations \eqref{eqn:B comm. relations}, \eqref{eqn:B right|0>}, \eqref{eqn:B left<0|}, and initial values $\langle0|0\rangle=1, \langle0|\phi_0|0\rangle=0$.
	This paring is called the vacuum expectation value.
	The following special vacuum expectation values are easily obtained
	\begin{align*}
		\langle0|\phi_i\phi_j|0\rangle=
		\begin{cases}
			\frac{1}{2},\ \ \ \ &i=j=0,\\
			1,\ \ \ \ &-i=j>0,\\
			0,\ \ \ \ &others.
		\end{cases}
	\end{align*}
	In general, there is also the Wick theorem to calculate general vacuum expectation values.
	
	The normal-ordering product of neutral fermions is defined by
	\begin{align}\label{eqn:B ::}
		:\phi_i\phi_j:=
		\phi_i\phi_j-\langle0|\phi_i\phi_j|0\rangle.
	\end{align}
	It is also convenient to introduce the following generating series of neutral fermions,
	\begin{align*}
		\phi(z):=\sum_{i\in\mathbb{Z}} \phi_i z^i.
	\end{align*}
   Applying \eqref{eqn:B ::} we have the following OPE:
   \begin{align}
   \phi(u)\phi(v)=:\phi(u)\phi(v):+\frac{u-v}{2(u+v)}
   \end{align}
   where $\frac{u-v}{2(u+v)}$ is expanded in the region $|u|>|v|$. See \cite{BH,WY} for the details.

	\subsection{B-type boson-fermionic correspondence and BKP hierarchy}
	The B-type bosons $H^B_n$ are defined by
	\begin{align}\label{eqn:H_n}
		H^B_n:=\frac{1}{2}\sum_{i\in\mathbb{Z}}(-1)^{i+1}:\phi_i\phi_{-i-n}:
	\end{align}
	for $n\in\mathbb{Z}_{odd}$.
	Define
	\begin{align*}
		H^B_+(\hat{\mathbf{t}}) = \sum_{n\in\mathbb{Z}_{+,odd}} t_n H^B_n,
	\end{align*}
	where $\hat{\mathbf{t}}=(t_1,t_3,...)$ is a family of formal variables.
	Then when restricting to the subspace $\mathcal{F}_B^{(0)}$,
	the famous B-type boson-fermion correspondence can be stated as follows.
	\begin{prop}
		There is an isomorphism between vector spaces
		\begin{align*}
			\sigma_B: \ \ \mathcal{F}_{B}^{(0)} \ \ &\rightarrow \ \ \mathbb{C}[\![\hat{\mathbf{t}}]\!]\\
			|V\rangle \ \ &\mapsto \ \ \langle0| e^{H^B_+(\hat{\mathbf{t}})}|V\rangle.
		\end{align*}
	\end{prop}

	About the BKP hierarchy,
	there is also a kind of realization of the B-type fermionic Fock space via the infinite wedge and the isotropic Grassmannian in literatures (see, for example, \cite{BH,JWY}).
	For simplicity, we restrict ourselves to the bosonic and fermionic spaces and ignore B-type infinite wedge and isotropic Grassmannian.
	In this paper, we focus on the tau-functions of the BKP hierarchy in the big cell.
	They are tau-functions satisfying $\tau({\bf 0})\neq0$.
	Writing the corresponding vector $|V\rangle$ in the B-type fermionic Fock space, the above condition is equivalent to $\langle0|V\rangle\neq0$.
	
	A useful representation of $|V\rangle$ in the big cell is the following (see \cite{BH,WY})
	\begin{align}\label{eqn:a V BKP}
		|V\rangle=c\cdot e^A|0\rangle,
	\end{align}
	where $c$ is a non-zero constant, and
	\[A=\sum_{n,m\geq0}a_{n,m}\phi_m\phi_n,\]
	such that $a_{n,m}=-a_{m,n}$.
	This kind of expression of a given $|V\rangle$ is also unique.
	Those canonical numbers $\{a_{n,m}\}$ are called the B-type affine coordinates of the corresponding vector or, say, the tau-function $|V\rangle$.
	We notice that the non-zero constant $c$ in equation \eqref{eqn:a V BKP} will not cause more difficulty in results in this paper.

	\section{$W_{1+\infty}$ and $W^B_{1+\infty}$ algebras}\label{walgebra}
	In this section, we review the notion of $w_{1+\infty}, w^B_{1+\infty}$ algebras and their realizations via fermions.
	We recommend the references \cite{FKN, L} for interested readers.
	
	\subsection{$w_{1+\infty}$ and $W_{1+\infty}$ algebras}
	Denote by $w_{1+\infty}$ the algebra of diffeomorphism on the circle,
	\begin{align}
		w_{1+\infty}:=\text{span}_{\mathbb{C}}\{z^i\partial_z^j\ |i\in\mathbb{Z}, j\in\mathbb{Z}_{\geq0}\}.
	\end{align}
	The algebra $w_{1+\infty}$ has a natural degree by assigning $\deg z=1$ and $\deg \partial_z=-1$.
	For each operator $a$, we define its adjoint $a^*$ in terms of
	\begin{align}\label{eqn:a^*}
		(z^k\partial_z^l)^* := (-\partial_z)^l z^k,\ \ \
		k\in\mathbb{Z}, l\in\mathbb{Z}_{\geq0}.
	\end{align}
	There is a realization of a central extension of $w_{1+\infty}$ algebra via the fermions, which is the $W_{1+\infty}$ algebra.
	This realization has many different choices.
	In this paper, we fix it as
	\begin{equation}\label{eww}
		\begin{split}
			\widehat{} \ : \  w_{1+\infty}&\longrightarrow W_{1+\infty}\\
			a&\longrightarrow \hat a : = \Res_z(:\psi^*(z)a\cdot\psi(z):).
		\end{split}
	\end{equation}
	There are natural Lie algebra structures on $w_{1+\infty}$ and $W_{1+\infty}$ respectively.
	Together with the fermionic Fock space,
	this gives a projective representation of $w_{1+\infty}$.
	
	\begin{lem}\label{thm_wpsicom}
		The elements in $W_{1+\infty}$ algebra have the following commutation relations with fermionic fields:
		\begin{align}
			[\widehat{a},\psi^*(z)] = a^*\cdot\psi^*(z),\quad [\widehat{a},\psi(z)] = -a\cdot\psi(z).
		\end{align}
	\end{lem}
	\begin{proof}
		We only verify the second equality here, since the first equality is obtained from a parallel argument.
		Without losing generality,
		we take $a=\sum_{i\in \mathbb{Z}\atop j> 0}a_{ij}z^i\partial_z^j$. Then from the definition of fermionic field \eqref{defferfield} and normal ordering \eqref{def_normalorder},
		we obtain
		\be
		\begin{aligned}
			\widehat{a}&=\Res_z(:\psi^*(z)a\cdot\psi(z):) \\
			&=\sum_{\mbox{$\tiny\begin{array}{c}
						i\in\mathbb{Z}, j>0 \\
						k, l\in\mathbb{Z}+\half  \\
						\atop k-l+i-j=0
					\end{array}$}}a_{ij}[k-\frac{1}{2}]_{j}
			\cdot\psi_l^*\psi_k
			-\sum_{\mbox{$\tiny\begin{array}{c}
						i\in\mathbb{Z}, j>0 \\
						k, l\in\mathbb{Z}+\half  \\
						\atop k-l+i-j=0\\
						k<0
					\end{array}$}}a_{ij}[k-\frac{1}{2}]_{j}
			\cdot\textbf{1},
		\end{aligned}
		\ee
		where $[k-\frac{1}{2}]_{j}:=\prod_{m=0}^{j-1} (k-\frac{1}{2}-m)$.
		Then, from the anti-commutation relations \eqref{eqn:fercomm} of fermions, we have
		\be
		[\widehat{a},\psi(z)]=-\sum_{\mbox{$\tiny\begin{array}{c}
					i\in\mathbb{Z}, j>0 \\
					k, l \in\mathbb{Z}+\half  \\
					\atop k-l+i-j=0
				\end{array}$}}a_{ij}[k-\frac{1}{2}]_{j} z^{l-1}\psi_k.
		\ee
		By substituting $l=k-j+i$, one can see that it is identical to $-a\cdot \psi(z)$.
	\end{proof}

	The $W_{1+\infty}$ algebra has Heisenberg algebra, and Virasoro algebra as its subalgebras.
	After boson-fermionic correspondence, these subalgebras make important roles in many mathematical physics models.
	\begin{ex}\label{exa:w-W}
		We list some useful subalgebras with generators given in the following
		$${\renewcommand{\arraystretch}{1.5}
			\begin{array}{|c|c|c|c|}\hline
				\hbox{\small Subalgebra}&	\hbox{\small $w_{1+\infty}$}  & \hbox{\small $W_{1+\infty}$ }  \\\hline
				\text{Heisenberg\ algebra} & j_n:= -z^{n}  &  \widehat{j_n}:=J_n \\\hline
				\text{Virasoro\ algebra} &  l_n:= -z^{n}\left(z\frac{\partial}{\partial z} +\frac{n+1}{2}\right)  & \widehat{l_n}:=L_n\\\hline
		\end{array}}$$
		where $J_n$ are given in equation \eqref{eqn:J_n}, and
		\begin{align}
			L_n=-\sum_{l,k\in\mathbb{Z}+\frac{1}{2}\atop l-k=n} (k+\frac{n}{2}):\psi^*_l \psi_k:.
		\end{align}
		For convenience, we use the same notation of the above operators and their counterparts via boson-fermion correspondence.
		Some detailed derivation of this transformation can be seen in \cite{FKN,LY3} based on vertex operators.
		Using notations
		\begin{align*}
			\alpha_n = \frac{\partial}{\partial t_n},\ \ \ \
			\alpha_{-n} = nt_n,\ \
			\text{\ for\ \ }n\in\mathbb{Z}_{>0},
		\end{align*}
		and $\alpha_0=0$, the above operators are given by
		\begin{align}
			J_n=\alpha_n,\ \ \ \ \ \ \
			L_n=\half\sum_{i\in \mathbb{Z}}\bm{:}\alpha_i\alpha_{n-i}\bm{:},
		\end{align}
		where $\bm{:}\cdot\bm{:}$ is the bosonic normal ordering, which moves $\alpha_n, n>0$ to the right.
	\end{ex}

	\subsection{$w^B_{1+\infty}$ and $W^B_{1+\infty}$ algebras}
	For the B-type case,
	let $\iota$ be an involution on $w_{1+\infty}$ in terms of
	\begin{align}
		\iota\big(z^k(z\partial_z)^m\big):=(-z\partial_z)^m (-z)^k.
	\end{align}
	Then, one can obtain the useful subalgebra $w^B_{1+\infty}$ as
	\begin{align}
		w^B_{1+\infty}:=\{l_z\in w_{1+\infty}\ \ |\ \ \iota(l_z)=-l_z\}.
	\end{align}
	
	For any given differential operator $l\in w_{1+\infty}$, one can define the associated B-type fermionic operators as (see \cite{L,A23})
	\begin{align}
		\hat{l}^B:=\frac{1}{2}\text{Res}_{w}\ w^{-1}\cdot :\phi(z)\ l_w\cdot \phi(w):|_{z=-w}.
	\end{align}
	The image of this realization is the $W^B_{1+\infty}$ algebra.
	This gives a realization of $w^B_{1+\infty}$.
	\begin{lem}\label{lem:B hatl-l}
		For any $l\in w_{1+\infty}$,
		\begin{align*}
			[\hat{l}^B,\phi(z)]
			=-\frac{1}{2}\big(l_z-\iota(l_z)\big) \cdot \phi(z).
		\end{align*}
	\end{lem}
	{\bf Proof:}
	The proof of this Lemma is similar to Lemma \ref{thm_wpsicom}.
	One only needs to replace the anti-commutation relations for charged fermions to neutral fermions in equation \eqref{eqn:B comm. relations}.
	$\Box$
	
	\begin{defn} \label{tildel}
For a given differential operator $l\in w_{1+\infty}$, it is convenient to define its anti-symmetrization
\begin{align}
	\tilde{l}:=\frac{1}{2}\big(l-\iota(l)\big).
\end{align}
	\end{defn}
	
	As a corollary, we obtain the following:
	\begin{cor}
		For any $l\in w_{1+\infty}$, denote by $\tilde{l}$ its anti-symmetrization with respect to $\iota$, then
		$$[\hat{l}^B,\phi(z)]=[\hat{\tilde{l}}^B,\phi(z)].$$
	\end{cor}
	{\bf Proof:}
	Both of the two sides of the above equation are equal to
	$-\tilde{l}\cdot\phi(z)$
	since $\tilde{\tilde{l}}=\tilde{l}$.
	$\Box$
	
	\begin{ex}\label{exa:B w-W}
		We list some useful operators in $w^B_{1+\infty}$ and their realizations in $W^B_{1+\infty}$.
		$${\renewcommand{\arraystretch}{1.5}
			\begin{array}{|c|c|c|c|}\hline
				\hbox{\small Subalgebra}&	\hbox{\small $w^B_{1+\infty}$}  & \hbox{\small $W^B_{1+\infty}$ }  \\\hline
				\text{Heisenberg} & j_n:= -z^{n}, n \text{\ odd\ }  &  \hat{j}_n^B:=H_n^B \\\hline
				\text{Virasoro} &  l_n:= -z^{n}\left(z\partial_z+\frac{n}{2}\right), n \text{\ even\ }  & \hat{l}_n^B:=L_n^B\\\hline
		\end{array}}$$
		Here $H^B_n$ are given in equation \eqref{eqn:H_n} and
		\begin{align}
			L_n^B=\frac{1}{2}\sum_{i\in\mathbb{Z}}(-1)^{i}i:\phi_{-i-n}\phi_i:, \ n\in\mathbb{Z}_{even}.
		\end{align}
		After boson-fermionic correspondence (see, for example, \cite{L,A23,LY3}),
		they are given by
		\begin{align}
			H_n^B=& \begin{cases}
				\frac{\partial}{\partial t_n}, &\ n\in\mathbb{Z}_{+,odd},\\
				\frac{-n}{2}t_{-n}, &\ n\in\mathbb{Z}_{-,odd},
			\end{cases}\\
			L_n^B=&
			\sum_{i+j=n}\bm{:}J_i^B J_j^B\bm{:}, \ n\in\mathbb{Z}_{even}.
		\end{align}
	\end{ex}

	\section{Lifting operator and fermionic two-point function}\label{mainresult}
	In this section, we introduce the notion of the lifting operator and derive our main results.
	
	\subsection{KP case}\label{sec:main A}
	For a given fermionic tau-function $\V$ of the KP hierarchy in the big cell of Sato's Grassmannian $Gr_{(0)}$,
	we define the fermionic one-point function and  its dual  by
	\begin{align}\label{defnone}
		\Psi(z) =& \frac{\langle 1|  \psi(z) \V} {\langle 0  \V}=  \frac{\tau(\bt)|_{t_k=-z^{-k}/k }}{\tau(\bt)|_{t_k=0}},\\
		\Psi^*(z) =& \frac{ \langle -1|  \psi^*(z) \V } {\langle 0  \V}=  \frac{\tau(\bt)|_{t_k=z^{-k}/k }}{\tau(\bt)|_{t_k=0}},
	\end{align}
	where $\langle1|=\bra{0}\psi^*_{\half}$ and $\langle-1|=\bra{0}\psi_{-\half}$.
	The second equalities in both expressions above are direct consequences of boson-fermion correspondence (see Section 6.3 in \cite{DJM}).
	Thus, $\Psi(z)$ and $\Psi^*(z)$ are also called the wave function and dual wave function of the corresponding tau-function.

	The fermionic one-point function is also closely relevant to the first canonical basis vector of $\ket{V}$.
	If $|V\rangle$ is given by the following formalism
	$$
	\ket{V} = c\cdot \tilde{\varphi}_{0}\wedge\tilde{\varphi}_{1}\wedge\tilde{\varphi}_{2}\wedge \cdots
	$$
	where $\tilde{\varphi}_i$ is the canonical basis of $|V\rangle$.
	Based on the equation \eqref{eqn:bki} for affine coordinates $b_{k,i}$,
	and equation \eqref{canobasis} for canonical basis,
	we have
	\begin{align} \label{vpk}
		\begin{split}
		\tilde\varphi_0(z) =& z^{\frac{1}{2}}+\frac{\braket{0| \psi^*_{k+\half}\sum_{i=0}^{\infty}z^{-i-\frac{1}{2}} \psi_{-i-\half}|V}}{\braket{0|V}}\\
		=& z^{\frac{1}{2}}\cdot \frac{\braket{0|\psi^*_{\half}\psi(z)|V}}{\braket{0|V}}
		=z^{\frac{1}{2}}\cdot \Psi(z).
		\end{split}
	\end{align}
Hence
$$
\Psi(z) = 1+\sum_{i=0}^\infty  b_{0,i} z^{-i-1}.
$$
Similarly, we have	
	$$
\Psi^*(z) = 1-\sum_{i=0}^\infty  b_{i,0} z^{-i-1}.
$$
	
	The canonical fermionic form of $|V\rangle$ in equation \eqref{eqn:V A} is closely related to the following fermionic two-point function
	$$
	\Psi(u,v)  :=  \frac{\langle 0| \psi^*(u) \psi(v)  \V}{\langle 0  \V} = \frac{1}{u-v} \frac{\tau(\bt)|_{t_k=\frac{u^{-k}-v^{-k}}{k} }}{\tau(\bt)|_{t_k=0}} .
	$$
	where the second equality can be obtained from boson-fermion correspondence directly (see e.g. \cite{O02})
	and $\frac{1}{u-v}$ is expanded in the region $|u|>|v|$.
	Also, based on the equation \eqref{eqn:bki} for affine coordinates $b_{k,i}$, one has
	\begin{align}
		\Psi(u,v) = \frac{1}{u-v} + \sum_{i,j\geq 0} b_{i,j} u^{-i-1} v^{-j-1}.
	\end{align}
	
	One can see that $\Psi(z)$ and $\Psi^*(z)$ record the first row and first column of affine coordinates.
	The $\Psi(u,v)$ records all information of affine coordinates, thus it determines the corresponding tau-function $|V\rangle$.
	Some applications of using $\Psi(u,v)$ to study the corresponding tau-function were studied in \cite{O02, ADKMV}. Thus, it will be important to find a simple formula for $\Psi(u,v)$.
	
	Now, we introduce the essential ingredient in this paper, the lifting operator.
	\begin{defn}[Lifting operator]\label{def:lifting}
		For a fermionic tau-function $|V\rangle$ of the KP hierarchy,
		a \textit{lifting operator} $l_z\in w_{1+\infty}$ for it has the following property.
		There exists an admissible basis $\{\varphi_k\}_{k\in\mathbb{Z}_{\geq0}}$ corresponding to $|V\rangle$,
		and a sequence of complex numbers $a_j, j\in\mathbb{Z}_{\geq1}$ such that
		$$
		(z^{\frac{1}{2}} \cdot l_z \cdot z^{-\frac{1}{2}})^{k}
		\cdot \varphi_0(z)
		=a_k\cdot \varphi_k(z)
		$$
		for all $k\in\mathbb{Z}_{\geq1}$.
		Without loss of generality,
		we extra assume the non-degenerate condition for $l_z$,
		that is if the degree one part of $l_z$ is given by $\sum_{k\geq0} c_k \cdot (z\partial_z)^k z$,
		then $c_0=1$.
		The subscription $z$ here is used to indicate the indeterminate involved.
		\footnote{ The conjugation via $z^{\frac{1}{2}}$ here just comes from different definitions (see, for example, equation \eqref{vpk}).}
	\end{defn}
	
	Firstly, a lifting operator is naively a Kac--Schwarz operator \cite{KS}, i.e., for each $k\in\mathbb{Z}_{\geq0}$,
	$$
	(z^{\frac{1}{2}} \cdot l_z \cdot z^{-\frac{1}{2}}) \cdot\varphi_k \in \ket{V}.
	$$
	Thus, it is well-known that there exists a constant $c_l$ such that (see Lemma 3.2 in \cite{FKN})
	\begin{align}\label{eqn:l action = c}
		\hat{l} \cdot |V\rangle = c_l \cdot |V\rangle.
	\end{align}

	Moreover, the lifting operator for a fermionic-tau function should have the following property.
	\begin{lem}\label{lzform}
		For $\V\in Gr_{(0)}$, any lifting operator of it has the following form
		\begin{equation*}
			l_z = \mathrm{(degree\ one\ contribution)}
			+\mathrm{(non\text{-}positive\ degree\ contribution)},
		\end{equation*}
		where the non-positive degree contribution is an operator linearly generated by
		$z^k\partial_z^l$ with $k\leq l$,
		and the degree one contribution is of the following form: $z+\sum_{k\geq1} c_k \cdot (z\partial_z)^k z$
		with some complex numbers $c_k$ for $k\in\mathbb{Z}_{\geq1}$.
	\end{lem}
	\begin{proof}
		By the definition of the lifting operator,
		the degree of each term in $l_z$ should be less than or equal to one.
		Thus, the positive degree contribution $l_z^+$ of $l_z$ must be of degree one and can be formally written as
		\begin{align}\label{eqn:l_z as zdz}
			l_z^+ =z+\sum_{k\geq1} c_k \cdot (z\partial_z)^k z.
		\end{align}
	\end{proof}

\begin{lem}\label{lem:<0|l}
	For a lifting operator $l_z$,
	we have
	\begin{align}\label{eqn:<0|l}
		\bra{0}\hat{l} = \bra{0}\psi^*_{\half}\psi_{-\half}.
	\end{align}
\end{lem}
\begin{proof}
	First,
	from the $W_{1+\infty}$ realization of $z$ listed in Example \ref{exa:w-W},
	we have
	$$
	\bra{0}\hat{z} = -\bra{0}J_1 = -\bra{0}\psi_{-\half}\psi^*_{\half},
	$$
	which is equal to the right hand side of equation \eqref{eqn:<0|l}.
	Thus,
	from Lemma \ref{lzform},
	we only need to show that,
	for $k\in\mathbb{Z}_{\geq1}$,
	\begin{align*}
		\bra{0}\widehat{(z\partial_z)^k z}
		=0,
	\end{align*}
	and for an operator $l'$ with non-positive degree, $\bra{0}\widehat{l'}
	=0$.
	The latter trivially follows from the normal ordering in equation \eqref{eww} when defining the $W_{1+\infty}$ realization.
	For the former,
	\begin{align*}
		\bra{0} \widehat{(z\partial_z)^k z}
		=\sum_{a\in\mathbb{Z}}
		a^k \bra{0}:\psi^*_{a+\half} \psi_{a-\half}:
	\end{align*}
	is always equal to zero since, for $a\in\mathbb{Z}_{>0}$,
	$\psi_{a-\half}$ is a left annihilator,
	for $a\in\mathbb{Z}_{<0}$,
	$\psi^*_{a+\half}$ is a left annihilator,
	and for $a=0$,
	the coefficient $a^k$ vanishes.
\end{proof}

	The importance of introducing the lifting operator lays on the fact that it can be quite an efficient tool for calculating the fermionic two-point function via our main result Theorem \ref{thm:intro A}.
	We restate it more concretely as the following
	\begin{thm}[= Theorem \ref{thm:intro A}]
		\label{thm:main KP}
		Let $l_z$ be a lifting operator for the fermionic tau-function $\V$ of the KP hierarchy, as defined in Definition~\ref{def:lifting}. Denote the fermionic one-point functions $\Psi(v)$ and $\Psi^*(u)$,  as well as the fermionic two-point function  $\Psi(u,v)$ as those defined in Section~\ref{sec:main A}.
		 Then we have
		\begin{equation}\label{keyequation}
			(l_u^*-l_v) \cdot \Psi(u,v)  = \Psi^*(u)\Psi(v).
		\end{equation}
	\end{thm}
	\begin{proof}
		From Lemma \ref{thm_wpsicom}, we have
		\begin{equation}
			\begin{split}
				l_u^*\cdot\Psi(u,v) &= \frac{\langle 0| l^*_u\cdot\psi^*(u) \psi(v)  \V}{\langle 0  \V}
				= \frac{\langle 0|[\hat{l},\psi^*(u)] \psi(v)  \V}{\langle 0  \V}\\
				&=\frac{\langle 0|\hat{l}\psi^*(u) \psi(v) -\psi^*(u)\hat{l}\psi(v) \V}{\langle 0  \V}.
			\end{split}
		\end{equation}
		On the other hand, still by Lemma \ref{thm_wpsicom},
		\begin{align*}
			l_v\cdot\Psi(u,v) =& \frac{\langle 0|\psi^*(u) l_v\cdot\psi(v)  \V}{\langle 0  \V}\\
			=&  \frac{\langle 0|-\psi^*(u)\ \hat{l}\ \psi(v)
				+ \psi^*(u) \psi(v)\ \hat{l}\ \V}{\langle 0  \V} .
		\end{align*}
		Then,
		combining the above two equations,
		we have
		\begin{align}\label{eqn:l*-l action}
			(l_u^*-l_v)\Psi(u,v)=\frac{\langle 0|\hat{l}\ \psi^*(u)\psi(v)
				- \psi^*(u)\psi(v)\ \hat{l}\  \V}{\langle 0  \V}.
		\end{align}

		Recall that in Lemma \ref{lem:<0|l},
		we have proved
		$$
		\bra{0}\hat{l}= \bra{0}\psi^*_{\half}\psi_{-\half}.
		$$
		Thus, together with equation \eqref{eqn:l action = c} and equation \eqref{eqn:l*-l action}, we obtain
		\begin{align}\label{eqn:4 fermions}
			\begin{split}
			(l_u^*-l_v+c_l)\Psi(u,v) =& \frac{\braket{0|\psi^*_{\half}\psi_{-\half}\psi^*(u)\psi(v)|V}}{\braket{0|V}}\\
			=&\frac{\braket{0|\psi^*_{\half}\psi_{-\half}\psi^*(u)\psi(v)|V}}{\braket{0|V}}.
			\end{split}
		\end{align}
		Since $|V\rangle$ has the canonical formalism $\V=c\cdot e^A|0\rangle$ in equation \eqref{eqn:V A}, and $\langle0|e^{-A}=\langle0|$,
		we can insert $e^{A} e^{-A}$ to each fermion in the right hand side of the above equation \eqref{eqn:4 fermions}.
		By using Baker--Campbell--Hausdorff formula,
		each $e^{-A} \tilde{\psi} e^A$ is a linear summation of fermions.
		Thus, we can apply the Wick theorem.
		As a result, the right hand side of equation \eqref{eqn:4 fermions} is equal to
		\begin{align*}
			&\frac{\braket{0|\psi^*_{\half}\psi_{-\half}|V} \braket{0|\psi^*(u)\psi(v)|V}
				+\braket{0|\psi^*_{\half}\psi(v)|V}\braket{0|\psi_{-\half}\psi^*(u)|V}}
			{\braket{0|V}}\\
			=& b_{0,0}\Psi(u,v)+\Psi^*(u)\Psi(v),
		\end{align*}
		where we have used the equation \eqref{eqn:bki}.
		Therefore, we obtain
		$$
		(l_u^*-l_v+c_l-b_{0,0})\Psi(u,v) = \Psi^*(u)\Psi(v).
		$$
		Moreover,
		from the definition equation \eqref{eqn:l action = c} of the constant $c_l$ and equation \eqref{eqn:<0|l},
		we have
		\begin{align*}
			c_l
			=\frac{\langle0|\cdot \hat{l} \cdot |V\rangle}{\langle0|V\rangle}
			=\frac{\langle0|\psi^*_{\half}\psi_{-\half} |V\rangle}{\langle0|V\rangle}
			=b_{0,0}.
		\end{align*}
		Thus,
		this theorem is proved.
	\end{proof}

	\begin{rmk}
	The above Theorem \ref{thm:main KP} provides a powerful constraint for the fermionic two-point function $\Psi(u,v)$.
	In general,
	it will totally determine $\Psi(u,v)$ and provides a compact formula (see Section \ref{examples} for applications in some examples).

	\end{rmk}

	\subsection{BKP case}\label{sec:main B}
	For a tau-function $|V\rangle$ of the BKP hierarchy in the big cell, we can define the following functions
	\begin{defn}
		The first and second fermionic one-point functions are defined by
		\begin{align*}
			\Psi^B(z):=\frac{\langle0|\phi_0\phi(z)|V\rangle}{\langle0|V\rangle},
			\quad \tilde{\Psi}^B(z):=\frac{\langle0|\phi_{-1}\phi(z)|V\rangle}{\langle0|V\rangle}.
		\end{align*}
		The fermionic two-point function is defined by
		\begin{align*}
			\Psi^B(u,v):=\frac{\langle0|\phi(u)\phi(v)|V\rangle}{\langle0|V\rangle}.
		\end{align*}
	\end{defn}
	
	By fermionic computations (see Section 3 in \cite{WY}),
	those functions are some generating series of affine coordinates $a_{m,n}$ for $|V\rangle$.
	Concretely, we have the following
	\begin{align}\label{eqn:B psi-a}
		\Psi^B(z)=&
		\frac{1}{2}
		+\sum_{n>0}(-1)^{n}a_{0,n}z^{-n},\\
		\tilde{\Psi}^B(z)=&
		-z-a_{1,0}-2\sum_{n>0}(-1)^{n}a_{1,n}z^{-n},
	\end{align}
	and
	\begin{align}
		\begin{split}
		\Psi^B(u,v)=&-2\sum_{n,m>0}(-1)^{m+n+1}a_{n,m}u^{-n}v^{-m}\\
		&+\sum_{n>0}(-1)^na_{n,0}(u^{-n}-v^{-n})
		+\frac{u-v}{2(u+v)},
		\end{split}
	\end{align}
	where $\frac{u-v}{2(u+v)}$ is expanded in the region $|u|>|v|$.
	As a consequence,
	one can notice that, $\Psi^B(z)$ and $\tilde{\Psi}^B(z)$ record the first two rows of affine coordinates.
	$\Psi^B(u,v)$ records all information of affine coordinates.
	Some applications of using $\Psi^B(u,v)$ to study the corresponding tau-function were studied in \cite{WY}.
	
	Motivated by the last subsection and Lemma \ref{lzform}, we formulate the definition of the lifting operator of the B-type.
	\begin{defn}[Lifting operator of B-type]
		\label{def:lifting B}
		For a given fermionic tau-function $|V\rangle$ of the BKP hierarchy in the big cell, $l\in w_{1+\infty}$ is called a lifting operator with respect to $|V\rangle$, if $l$ is a Kac--Schwarz operator of $|V\rangle$ and it is of the following form
		\begin{align*}
			l=\mathrm{(energy\ one\ contribution)}
			+\mathrm{(non-positive\ energy\ contribution)}.
		\end{align*}
		Here, $l$ is a Kac--Schwarz operator means that there exists a constant $c_l$ such that
		\begin{align}\label{eqn:l = c B}
			\hat{l}^B \cdot |V\rangle = c_l \cdot |V\rangle,
		\end{align}
		where $\hat{l}^B$ is the realization of $l$ in $W^B_{1+\infty}$,
		and non-positive energy contributions are operators linearly generated by
		$z^k\partial_z^l$ with $k\leq l$.
		We extra assume the coefficient of the term $z$ in $l$ is one, similar to the KP hierarchy case.
	\end{defn}
	
	\begin{lem}
		Supposing $l$ is a lifting operator, then
		\begin{align}\label{eqn:<0|lB}
			\langle0|\hat{l}^B=\langle0|\phi_0\phi_{-1}.
		\end{align}
	\end{lem}
	{\bf Proof:}
	The proof is similar to Lemma \ref{lem:<0|l}.
	We only need to notice that
	\[\langle0|\hat{l}^B
	=\langle0|\hat{z}^B
	=-\langle0|H_1=\langle0|\phi_0\phi_{-1}.\]
	$\Box$

	Next, we prove our main result of B-type, Theorem \ref{thm:intro B main}.
	We restate it as follows.
	\begin{thm}[= Theorem \ref{thm:intro B main}]
		\label{thm:B main}
			Let $l_z$ be a lifting operator for the fermionic tau-function $\V$ of the BKP hierarchy, as defined in Definition~\ref{def:lifting B}. Let $\tilde{l}_z$ represent its anti-symmetrization, as specified in  \eqref{tildel}. Denote the BKP fermionic one-point functions $\Psi^B(v)$ and $\tilde \Psi^B(u)$,  as well as the fermionic two-point function  $\Psi^B(u,v)$ as those defined in Section~\ref{sec:main B}. Then we have
		\begin{align}
			\label{eqn:BKPmain}
			(\tilde{l}_u+\tilde{l}_v)\cdot\Psi^B(u,v)
			=\Psi^B(u)\tilde{\Psi}^B(v)
			-\Psi^B(v)\tilde{\Psi}^B(u).
		\end{align}
	\end{thm}
	{\bf Proof:}
	We assume $\langle0|V\rangle=1$ for convenience.
	On the one hand, by Lemma \ref{lem:B hatl-l},
	\begin{align}
		\begin{split}
		\tilde{l}_u\cdot\Psi^B(u,v)
		=&\langle0|\tilde{l}_u\cdot\phi(u)\phi(v)|V\rangle
		=-\langle0|[\hat{l}^B,\phi(u)]\phi(v)|V\rangle\\
		=&-\langle0|\hat{l}^B\phi(u)\phi(v)-\phi(u)\hat{l}^B\phi(v)|V\rangle.
		\end{split}
	\end{align}
	On the other hand, still by Lemma \ref{lem:B hatl-l},
	\begin{align}
		\begin{split}
		\tilde{l}_v\cdot\Psi^B(u,v)
		=&\langle0|\phi(u)\tilde{l}_v\cdot\phi(v)|V\rangle
		=-\langle0|\phi(u)[\hat{l}^B,\phi(v)]|V\rangle\\
		=&-\langle0|\phi(u)\hat{l}^B\phi(v)-\phi(u)\phi(v)\hat{l}^B|V\rangle.
		\end{split}
	\end{align}
	Thus, by combining the above two equations,
	\begin{align*}
		(\tilde{l}_u+\tilde{l}_v)\cdot\Psi^B(u,v)
		=-\langle0|\hat{l}^B\phi(u)\phi(v)-\phi(u)\phi(v)\hat{l}^B|V\rangle.
	\end{align*}
	Since $l$ is a lifting operator of $|V\rangle$, we have $\hat{l}^B \cdot |V\rangle=c_l \cdot |V\rangle$ and from equation \eqref{eqn:<0|lB},
	\begin{align}
		\langle0|\hat{l}^B|V\rangle
		=\langle0|\phi_0\phi_{-1}|V\rangle=a_{1,0}.
	\end{align}
	Then, we have
	\begin{align}
		(\tilde{l}_u+\tilde{l}_v-c_l)\cdot\Psi^B(u,v)
		=&-\langle0|\phi_0\phi_{-1}\phi(u)\phi(v)|V\rangle.
	\end{align}
	By the Wick theorem, the right hand side of the above equation is equal to
	\begin{align}
		\begin{split}
		-\langle0|\phi_0\phi_{-1}|V\rangle \langle0|\phi(u)\phi(v)|V\rangle
		+&\langle0|\phi_0\phi(u)|V\rangle \langle0|\phi_{-1}\phi(v)|V\rangle\\
		-&\langle0|\phi_0\phi(v)|V\rangle \langle0|\phi_{-1}\phi(u)|V\rangle,
		\end{split}
	\end{align}
	which is further equal to
	\[-a_{1,0}\cdot\Psi^B(u,v)
	+\Psi^B(u)\tilde{\Psi}^B(v)
	-\Psi^B(v)\tilde{\Psi}^B(u).\]
	Thus,
	combining the above three equations,
	we have proved
	\begin{align}
		(\tilde{l}_u+\tilde{l}_v-c_l+a_{1,0})\cdot\Psi^B(u,v)
		=\Psi^B(u)\tilde{\Psi}^B(v)
		-\Psi^B(v)\tilde{\Psi}^B(u).
	\end{align}
	Finally,
	from equations \eqref{eqn:l = c B} and \eqref{eqn:<0|lB},
	we have
	\begin{align*}
		c_l
		=\langle0|\cdot \hat{l}^B \cdot |V\rangle
		=\langle0|\phi_0\phi_{-1} |V\rangle
		=a_{1,0}.
	\end{align*}
	This theorem is thus proved.
	$\Box$

	\section{Applications in KP and BKP hierarchies}\label{examples}
	
	In this section, we apply our main theorems to derive fermionic two-point functions for the generalized Kontsevich model and Br{\' e}zin--Gross--Witten model, which satisfy the KP hierarchy and BKP hierarchy, respectively.

	\subsection{KP application: the generalized Kontsevich model.}
	\label{GKM}
	
	The generalized Kontsevich matrix model is defined as a Hermitian matrix model with an external field (see \cite{KMMMZ}, and see also \cite{ACM, A21}).
	It is a generalization of the famous Kontsevich model, which was initially used by Kontsevich to prove Witten's conjecture (see \cite{K,W90}).
	When the potential function for the generalized Kontsevich model is a monomial,
	it is related to the Witten's $r$-spin model
	(see \cite{W,AM,FSZ}).
	
	In this paper, since we do not use the concrete formalism of matrix integral, we omit the definition of the generalized Kontsevich model via matrix integral for conciseness
	(see \cite{KMMMZ,ACM,A21} for more details).
	We mainly follow the notations in Section 3 in \cite{A21}.
	A generalized Kontsevich model is defined from a potential function $V(z)$. Here, we will only consider the polynomial potential.
	Then, the generalized Kontsevich model produces a tau-function of the KP hierarchy.
	Our main result in this section is to give a formula for the fermionic two-point function of the tau-function from the generalized Kontsevich model with general polynomial potential $V(z)$.
	This result will also be used by the authors to study the spectral curve of the generalized Kontsevich model in a separate paper \cite{GJYZ}.
	
	For convenience, we will also denote by $x(z)=V'(z)$, $U(z)=\frac{1}{\hbar}x'(z)$.
	Some Kac--Schwarz operators for the generalized Kontsevich model can be derived from the matrix model directly and are well-known.
	For example, the following two operators (see Lemma 3.3 in \cite{A21}) are Kac--Schwarz operators for the generalized Kontsevich model with the potential function $V(z)$
	\footnote{Note that our operators are in fact dual to operators used in \cite{A21}.
		They will not cause any essential problem and can be obtained just by letting $\hbar\rightarrow-\hbar$ from his notation.},
	\begin{equation}\label{eqn:KS GKM}
		\begin{aligned}
			l_V(z) &=&& z-\frac{1}{U(z)}\frac{\partial}{\partial z}+\frac{U'(z)}{2U(z)^2}, \\
			q_V(z) &=&& \frac{1}{\hbar}\big(x(l_V(z))-x(z)\big).
		\end{aligned}
	\end{equation}
	Here, the $U(z)$ in the denominator should be understood as a series near $z=\infty$.
	Moreover,
	$q_V(z)$ annihilates the fermionic one-point function of the generalized Kontsevich model.
	That is to say, one has
	\begin{align}\label{eqn:qsc}
		\big(x(l_V(z))-x(z)\big)\cdot \Psi(z)=0
	\end{align}
	and its dual version
	\begin{align}\label{eqn:qsc dual}
		\big(x(l_V^*(z))-x(z)\big)\cdot \Psi^*(z)=0.
	\end{align}
	
	One can also notice that
	the leading term of $l_V(z)$ is $z$.
	Thus, as a corollary,
	\begin{cor}
		$l_V$ is a lifting operator of the generalized Kontsevich model determined by the potential function $V$.
	\end{cor}
	
	With the help of the lifting operator, we can apply our main theorem to derive the fermionic two-point function for the generalized Kontsevich model.
	\begin{prop}
		\label{thm:GKM}
		For the generalized Kontsevich model with polynomial potential $V(z)$ and $x(z)=V'(z)$, the fermionic two-point function is given by:
		\begin{align}\label{eqn:GKM two-point}
			\begin{split}
			\Psi(u,v)=&\frac{W\big(l_V^*(u),l_V(v)\big)}{x(u)-x(v)}
			\cdot\Psi^*(u)\Psi(v)\\
			&\qquad\qquad\in\frac{1}{u-v}+u^{-1}v^{-1}\mathbb{C}[\![u^{-1},v^{-1}]\!],
			\end{split}
		\end{align}
		where $\Psi(v)$ and $\Psi^*(u)$ are the fermionic one-point and dual one-point functions of the generalized Kontsevich model,
		$$W(u,v):=\frac{x(u)-x(v)}{u-v}
		\in\mathbb{C}[u,v]$$
		is a polynomial in $u,v$ , and $l_V^*$ is the adjoint of the lifting operator $l_V$ given in equation \eqref{eqn:KS GKM}.
	\end{prop}
	
	\begin{proof}
		Denote the right hand side of equation \eqref{eqn:GKM two-point} by $\bar{\Psi}(u,v)$,
		by using equation in Theorem \ref{thm:main KP} and the uniqueness of solution of that equation with certain leading terms, one only needs to show that
		\begin{align}\label{eqn:GKM final}
			\big(l_V(u)^*-l_V(v)\big)\cdot \bar{\Psi}(u,v)
			=\Psi^*(u) \Psi(v).
		\end{align}
		
		Since
		$[l_V(v),x(v)]=\hbar$ and $[l^*_V(u),x(u)]=-\hbar$,
		we have
		\begin{equation}
			[l^*_V(u)-l_V(v), x(u)-x(v)]=0.
		\end{equation}
		Thus, we obtain
		\begin{align*}
			\big(l_V(u)^*-l_V(v)\big)\cdot \bar{\Psi}(u,v)
			=\frac{\big(x(l_V^*(u))-x(l_V(v))\big)}{x(u)-x(v)}
			\cdot\Psi^*(u)\Psi(v).
		\end{align*}
		By equations \eqref{eqn:qsc} and \eqref{eqn:qsc dual},
		the right hand side of the above equation equals to
		$\Psi^*(u) \Psi(v)$.
		Thus, the above equation is equivalent to equation \eqref{eqn:GKM final}
		and this proposition is proved.
	\end{proof}

	\begin{rmk}\label{rmk:ACM and A21}
		In \cite{ACM,A21},
		other formulas for the Cauchy--Baker--Akhiezer kernel of the generalized Kontsevich model were derived.
		Their formulas depend on complicated integrals (see equation (3.16) in \cite{ACM} and equation (3.31) in \cite{A21}).
		It will be interesting to compare their formulas and ours.
	\end{rmk}

	\subsection{KP application: the $r$-spin model}
	
	Witten \cite{W} introduced the $r$-spin model and conjectured that it is related to the $r$-KdV hierarchy.
	This conjecture has been proved in \cite{FSZ}.
	The partition function of the $r$-spin model is the generating function of certain intersection numbers over the moduli spaces of $r$-spin curves.
	Thus, this result establishes the deep connection between the geometry of moduli spaces and integrable systems.
	
	Since the $r$-KdV hierarchy is a reduction of the KP hierarchy,
	one can apply Theorem \ref{thm:intro A} established in this paper to study the $r$-spin model.
	The tau-function $\tau^{(r)}(\mathbf{t})$ of the $r$-spin model satisfies the following string equation (see \cite{W})
	\begin{align*}
		\Bigg(\sum_{n\geq1}(n+r)t_{n+r}\frac{\partial}{\partial t_n}
		+\frac{1}{2}\sum_{a+b=r}abt_at_b
		-\frac{1}{\hbar}\frac{\partial}{\partial t_1}\Bigg)
		\cdot \tau^{(r)}(\mathbf{t})=0.
	\end{align*}
	Thus, comparing to $W_{1+\infty}$ operators listed in Example \ref{exa:w-W}, a lifting operator of $\tau^{(r)}$ could be
	\begin{align}\label{eqn:l r-spin}
		l_z:=z-\hbar z^{-r}\big(z\partial_z+\frac{-r+1}{2}\big).
	\end{align}
	The corresponding constant in equation \eqref{eqn:l action = c} is $c_l=0$.
	The adjoint operator of $l_z$ is
	\begin{align}
		\begin{split}
		l_z^*=&z-\hbar
		\big(-\partial_z z^{-r+1}+\frac{-r+1}{2}z^{-r}\big)\\
		=&z+\hbar
		\big(z^{-r+1}\partial_z+\frac{-r+1}{2}z^{-r}\big).
		\end{split}
	\end{align}
	
	Comparing to the Kac--Schwarz operators for the generalized Kontsevich model in equation \eqref{eqn:KS GKM},
	one can see that,
	the $r$-spin model corresponds to the special case of the generalized Kontsevich model satisfying
	$V(z)=\frac{z^{r+1}}{r(r+1)}$.
	Thus, just by applying Proposition \ref{thm:GKM},
	we obtain a formula for the fermionic two-p  oint function of $r$-spin theory.
	\begin{cor}[= Corollary \ref{thm: intro r-spin}]
		\label{thm: r-spin}
		The fermionic two-point function for the $r$-spin model is given by
		\begin{align}\label{eqn:2-point for r-spin}
			\Psi(u,v)=\frac{\big(\sum_{a=0}^{r-1}(l_u^*)^a (l_v)^{r-1-a}\big) \cdot \Psi^*(u) \Psi(v)} {u^r-v^r},
		\end{align}
		where $\Psi(v)$ and $\Psi^*(u)$ are the fermionic one-point and dual one-point functions of the r-spin model,
		$l$ is the lifting operator given in equation \eqref{eqn:l r-spin} and $l^*$ is its adjoint.
	\end{cor}
	
	\begin{rmk}
		It is known that fermionic one-point function $\Psi(z)$ can be solved by the equations \eqref{eqn:qsc} and \eqref{eqn:qsc}.
		Thus, equation \eqref{eqn:2-point for r-spin} gives a simple formula for the fermionic two-point function of Witten's $r$-spin model explicitly.
		It is a generalization of the result in \cite{TW,O01,O02}, (see also \cite{Z15,DYZ} for other methods) for the $r=2$,
		in which case $\Psi(w,z)$ is closely related to the well-known Airy kernel.
	\end{rmk}
	
	\subsection{BKP application: Br{\' e}zin--Gross--Witten model}
	\label{sec:BGW}
	The BGW model was introduced by Br{\' e}zin, Gross, and Witten when studying the lattice gauge theory (see \cite{BG, GW}).
	The enumerative geometric interpretation of this model was conjecturally proposed in \cite{N17},
	which states that the tau-function of the BGW model $\tau_{BGW}(\hat{\mathbf{t}})$ is also a generating function of some certain intersection numbers over the moduli spaces of stable curves.
	This conjecture was proved in \cite{CGG}.
	Recently, the BGW model was intensively studied, especially in its BKP structure (see \cite{A23,LY2,LY3,WY}).
	The conclusion is that $\tau_{BGW}(\hat{\mathbf{t}}/2)$ is a hypergeometric tau-function of the BKP hierarchy.
	
	The first thing for applying Theorem \ref{thm:B main} is to find a lifting operator for this model.
	The Virasoro constraints for the BGW model were derived in \cite{GN}.
	The first Virasoro equation for $\tilde{\tau}_{BGW}(\hat{\mathbf{t}}):=\tau_{BGW}(\hat{\mathbf{t}}/2)$ can be written as
	\begin{align*}
		\Big(\frac{1}{2}\sum_{k\in\mathbb{Z}_{+,odd}}kt_{k}\frac{\partial}{\partial t_k}
		-\frac{1}{\hbar}\frac{\partial}{\partial t_1} +\frac{1}{16} \Big)
		\cdot \tilde{\tau}_{BGW}(\hat{\mathbf{t}})=0.
	\end{align*}
	Thus, by Example \ref{exa:B w-W}, the lifting operator for $\tilde{\tau}_{BGW}(\hat{\mathbf{t}})$ could be
	\begin{align}\label{eqn:l for BGW}
		l_z:=z-\frac{\hbar z}{2}\partial_z \in w^B_{1+\infty},
	\end{align}
	and corresponding $c_l=-\frac{\hbar}{16}$.
	Since $l\in w^B_{1+\infty}$, we obtain $\tilde{l}=l$ in this case.
	
	Recall that the first and second fermionic one-point functions in the BKP hierarchy are the generating series of the first two rows of B-type affine coordinates in equation \eqref{eqn:B psi-a}.
	And in the BGW model case, corresponding affine coordinates were originally derived in \cite{WY}.
	We can list the explicit expressions for $\Psi^B(z)$ and $\tilde{\Psi}^B(z)$ as follows.
	\begin{ex}
		The first and second fermionic one-point functions of $\tilde{\tau}_{BGW}(\hat{\mathbf{t}})$ are
		\begin{align*}
			\Psi^B(z)
			=&\frac{1}{2}+\sum_{n>0} (-1)^n \frac{((2n-1)!!)^2}{2^{3n+1}\cdot n!} \hbar^n z^{-n},\\
			\tilde{\Psi}^B(z)
			=&-z+\frac{\hbar}{16}-2\sum_{n>0}(-1)^n \frac{(n-1)((2n-1)!!)^2}{2^{3n+5}\cdot(n+1)n!} \hbar^{n+1} z^{-n}.
		\end{align*}
	\end{ex}
	Using the above explicit formulas, one can check the following lemma directly.
	\begin{lem}
		The relation between $\Psi^B(z)$ and $\tilde{\Psi}^B(z)$ is
		\begin{align}\label{eqn:BGW psi-tpsi}
			\tilde{\Psi}^B(z) = -2 (l_z+\frac{\hbar}{16}) \cdot \Psi^B(z),
		\end{align}
		where $l_z$ is the lifting operator in equation \eqref{eqn:l for BGW}.
		Moreover, $\Psi^B(z)$ satisfies
		\begin{align}\label{eqn:BGW l-equation}
			(l_z^2-zl_z+\frac{\hbar}{2}l_z+\frac{\hbar^2}{16}) \cdot \Psi^B(z)=0.
		\end{align}
	\end{lem}

	\begin{rmk}
		The equation \eqref{eqn:BGW l-equation} is related to the quantum spectral curve of the BGW model.
	\end{rmk}

	\begin{thm}
		The fermionic two-point function with respect to $\tilde{\tau}_{BGW}(\hat{\mathbf{t}})$ is
		\begin{align}\label{eqn:2-for-BGW}
			\Psi^B(u,v)=
			\frac{-2(-2l_u+2l_v+u-v)\Psi^B(u)\Psi^B(v)}{\hbar(u+v)}.
		\end{align}
	\end{thm}
	{\bf Proof:}
	For convenience, write the right hand side of equation \eqref{eqn:2-for-BGW} as $\bar{\Psi}^B(u,v)$. Then, by Theorem \ref{thm:B main}, to prove this proposition, one only needs to show that
	\begin{align}\label{eqn:BGW condition}
		(l_u+l_v)\cdot\bar{\Psi}^B(u,v)
		=\Psi^B(u)\tilde{\Psi}^B(v)
		-\Psi^B(v)\tilde{\Psi}^B(u)
	\end{align}
	for this model.
	
	Firstly, using the definition equation \eqref{eqn:l for BGW} of the lifting operator $l$, one directly obtains
	\begin{align*}
		(l_u+&l_v)\cdot\bar{\Psi}^B(u,v)\\
		=&\frac{-2(-2l_u^2+ul_u-\frac{\hbar u}{2} -vl_u+2l_v^2+ul_v-vl_v+\frac{\hbar v}{2})\Psi^B(u)\Psi^B(v)} {\hbar(u+v)}\\
		&-\frac{(-2l_u+2l_v+u-v)\Psi^B(u)\Psi^B(v)}{u+v}.
	\end{align*}
	Then, using equation \eqref{eqn:BGW l-equation} to simplify the $l^2$ terms in the right hand side of the above equation and using the relation \eqref{eqn:BGW psi-tpsi} between $\Psi^B(z)$ and $\tilde{\Psi}^B(z)$, one can obtain
	\begin{align*}
		(&l_u+l_v)\cdot\bar{\Psi}^B(u,v)\\
		=&\frac{-u\tilde{\Psi}^B(u)\Psi^B(v)+u\Psi^B(u)\tilde{\Psi}^B(v)
			-v\tilde{\Psi}^B(u)\Psi^B(v)+v\Psi^B(u)\tilde{\Psi}^B(v)}{(u+v)}.
	\end{align*}
	The above equation is equivalent to equation \eqref{eqn:BGW condition}.
	Thus, this proposition is proved.
	$\Box$
	
	\begin{rmk}
		The formula for B-type fermionic two-point function of BGW model was initially obtained by Z. Wang and the third-named author via a totally different method \cite{WY}
		 (see also \cite{WYZ} for a generalization to the generalized BGW model).
	\end{rmk}
	
	\subsection{Further examples} \label{futherexamples}
	
    Indeed, in addition to the examples presented above, there are numerous other theories where our method can be applied. For instance, as studied by Chen and the first-named author in  \cite{CG}, the fermionic two-point functions of orbifold Gromov--Witten theory of $\mathbb{P}[r]$ can also be deduced using Theorem \ref{thm:main KP}.
    In what follows, we present a compilation of several examples along with their crucial data. This data encompasses the lifting operators, fermionic one-point functions, and two-point functions. Here, we utilize the notation $\hat{p}= \hbar \frac{d}{dx}$. These examples have been widely studied in various literature, and we refer readers to find details in the references below.

\begin{table}[h]
\centering
\footnotesize
	\begin{tabular}{|c|c|c|c|c|}\hline
	\hbox{\small KP Model}&	\hbox{\small KP coor.}   &
	\hbox{\small Lifting Operator} \\\hline
	\text{Simple Hurwitz} & $z=e^{x}$ &   $ e^{\hat{p}} e^{ x}$ \\\hline
	\text{Framed One-leg  vertex}&  $z=e^{{x}}$
	 	&  $ e^{f \hat p} e^{ x}$ \\\hline
	\text{Gromov--Witten theory of $\mathbb P[r]$} &  $z=x$

	 & $Ad_{e^{\frac{x}{\hbar}\ln x-\frac{x}{\hbar}}} \big(e^{-\hat p }\big)$  \\\hline
    \text{Monotone Hurwitz} & $z=x$ &  ${x}(1+{x}\hat{p}+\hbar)$\\\hline
    \text{Grothendieck's dessins d'enfants} & $z=x$
    & ${x}+{x}\hat{p}$ \\\hline
    \hbox{\small BKP Model}&	\hbox{\small BKP coor.}   &
    \hbox{\small Lifting Operator} \\\hline
    \text{Spin Hurwitz numbers} & $z=e^{x}$
    & $e^{\frac{(\hat{p}-\hbar)^{r+1}-\hat{p}^{r+1}
    		+(-1)^r\hbar^{r+1}}{r+1}} e^x$ \\\hline
	\end{tabular}
\end{table}

\subsubsection{Simple Hurwitz numbers}
The Hurwitz numbers counts branched covers between Riemann surfaces.
It is known that generating function of simple Hurwitz numbers satisfies the KP hierarchy (see, for examples, \cite{O00,KL}).
The fermionic one-point function of this model and its dual can be obtained by using the fermionic operator form as:
$$
\Psi(z) =\sum\limits_{k=0}^\infty \frac{e^{(k^2-k) \hbar/2}}{\hbar^k  k!} z^{-k}  ,\qquad \Psi^*(z)=\sum\limits_{k=0}^\infty \frac{(-1)^ke^{-(k^2-k) \hbar/2}}{\hbar^k  k!} z^{-k}.
$$
In this case one can take the lifting operator $l_v$ and its adjoint operator $l_u^*$ as $e^{\hbar  v \partial_v} v$ and $e^{-u \hbar   \partial_u  }u  $, respectively.
Then it is direct to obtain the fermionic two-point function of this model
\begin{align*}
\qquad  \Psi(u,v) =
  \ & \frac{1}{u-v}+ \sum_{m,n\geq 0}\frac{(-1)^{n} e^{\frac{m(m+1)}{2}\hbar-\frac{n(n+1)}{2}\hbar}}{\hbar^{m+n+1}(m+n+1) m!n!}   u^{-n-1}  v^{-m-1}
\end{align*}
by solving the equation \eqref{keyequation} in our main Theorem \ref{thm:main KP} as:
\begin{align*}
 (l_u^*-l_v) \Psi(u,v) =& \frac{ue^{-\hbar}}{u e^{-\hbar} -v}+  \sum_{m,n\geq 0}   \frac{(-1)^{n}  e^{\frac{m(m+1)}{2}\hbar-\frac{n(n-1)}{2}\hbar}}{\hbar^{m+n+1}(m+n+1) m!n!} u^{-n}  v^{-m-1}\\
 &  - \frac{ve^\hbar}{u -v e^\hbar} -   \sum_{m,n\geq 0} \frac{(-1)^{n} e^{\frac{m(m-1)}{2}\hbar-\frac{n(n+1)}{2}\hbar}}{\hbar^{m+n+1}(m+n+1) m!n!} u^{-n-1}  v^{-m}\\
 =&   \Psi^*(u) \Psi(v)
\end{align*}
and matches with the result computed in \cite{KL}.

\subsubsection{Framed one-leg  vertex (mirror of $\mathbb{C}^3$)}
The topological vertex provides a way to study the open Gromov--Witten invariants of all smooth toric Calabi--Yau threefolds \cite{ADKMV}.
When restricting to one-leg vertx,
following  \cite{ADKMV, DZ}, the fermionic one-point function of this model and its dual of this model are given by
\begin{align*}
	\Psi(z)  =& \sum\limits_{k=0}^\infty \frac{e^{-\frac{\hbar}{4}(2f+1)k(k-1)}}{[k]!} z^{-k},\\
	\Psi^*(z)=&\sum\limits_{k=0}^\infty \frac{(-1)^{k}e^{-\frac{\hbar}{4}(2f+1)k(k-1)}}{[k]!} z^{-k}.
\end{align*}
where $[n]:=\sinh \frac{n \hbar}{2}  $ and $[n]!:=\prod_{j=1}^n \sinh\frac{j \hbar}{2}  $.
 In this case we have the lifting operator $l_v= e^{f\cdot\hbar v\partial_v} v$ and its adjoint operator $l_u^*=e^{-f\cdot\hbar u\partial_u} u$. Then, one can obtain the fermionic two-point function
$$
\Psi(u,v) =\frac{1}{u-v}+\sum_{m,n\ge 0}^{}\frac{(-1)^ne^{\frac{\hbar}{4}(2f+1)(m+n+1)(m-n)}}{[m+n+1][m]![n]!}u^{- n-1}v^{-m-1}
$$
by solving the equation \eqref{keyequation} as:
\begin{align*}
 &(l_u^*-l_v) \Psi(u,v) \\
 =& \frac{ue^{-f\hbar}}{u e^{-f\hbar} -v}+  \sum_{m,n\geq 0}  \frac{(-1)^ne^{\frac{\hbar}{4}2f(m+n)(m+n-1)+(m+n+1)(m-n)}}{[m+n+1][m]![n]!} u^{-n}  v^{-m-1}\\
 &  - \frac{ve^{f\hbar}}{u -v e^{f\hbar}} -   \sum_{m,n\geq 0} \frac{(-1)^ne^{\frac{\hbar}{4}2f(m(m+1)-n(n+3))+(m+n+1)(m-n)}}{[m+n+1][m]![n]!} u^{-n-1}  v^{-m}\\
 =&   \Psi^*(u) \Psi(v).
\end{align*}

\subsubsection{Gromov--Witten theory of $\mathbb P[r]$}  The fermionic representation and quantum curves for the Gromov--Witten theory of $\mathbb{P}^1$ and its orbifold generalization were studied in \cite{DMNPS, CG} and references therein.
More precisely,
The fermionic one-point function and its dual of this model are:
\begin{align*}
 \Psi(z) &=e^{\frac{z}{\hbar}\ln z-\frac{z}{\hbar}}\sum\limits_{k=0}^\infty
	 \frac{(-1)^k q^{rk}}{r^k k!} \frac{\hbar^{-\frac{z}{\hbar}-(r+1)k}}{\Gamma(\frac{z}{\hbar}+rk+\frac{1}{2})},\\
 \Psi^*(z) &=e^{-\frac{z}{\hbar}\ln z+\frac{z}{\hbar}}\sum\limits_{k=0}^\infty
	 \frac{q^{rk}}{r^k k!} \frac{\Gamma(\frac{z}{\hbar}-rk+\frac{1}{2})}{\hbar^{-\frac{z}{\hbar}+(r+1)k}}.
\end{align*}
The lifting operator and its adjoint can be chosen as (see \cite{CG}):
\begin{align*}
l_z:={e^{\frac{z}{\hbar}\ln z-\frac{z}{\hbar}}}
\big(e^{-\hbar\frac{d}{dz}}\big)
{e^{-\frac{z}{\hbar}\ln z+\frac{z}{\hbar}}},
\qquad l^*_z:={e^{-\frac{z}{\hbar}\ln z+\frac{z}{\hbar}}}
\big(e^{\hbar\frac{d}{dz}}\big)
{e^{\frac{z}{\hbar}\ln z-\frac{z}{\hbar}}}.
\end{align*}
Then the fermionic two-point function
\begin{align*}
	\Psi(u, v)=&\frac{\mathrm{e}^{\frac{v}{\hbar} \ln \frac{v}{\hbar}-\frac{v}{\hbar}}}{\mathrm{e}^{\frac{u}{\hbar} \ln \frac{u}{\hbar}-\frac{u}{\hbar}}}\sum_{d=1}^{\infty}\Bigg(\frac{q^{r d}}{d \cdot r^d \hbar^{(r+1) d}}\\
	&\qquad\qquad\cdot \sum_{k=0}^{d-1} \frac{(-1)^{k-1}}{k!(d-1-k)!} \sum_{n=1}^r \frac{\Gamma\left(\frac{u}{\hbar}+n-r(d-k)-\frac{1}{2}\right)}{\Gamma\left(\frac{v}{\hbar}+r k+n+\frac{1}{2}\right)}\Bigg),
\end{align*}
can be obtained by solving the equation \eqref{keyequation} as:
\begin{align*}
(l_u^*-l_v)\Psi(u,v)&=\frac{\mathrm{e}^{\frac{v}{\hbar} \ln \frac{v}{\hbar}-\frac{v}{\hbar}}}{\mathrm{e}^{\frac{u}{\hbar} \ln \frac{u}{\hbar}-\frac{u}{\hbar}}} \sum_{d=1}^{\infty}\Bigg(\frac{q^{r d}}{d \cdot r^d \hbar^{(r+1) d}}\\
	&\qquad\cdot \sum_{k=0}^{d-1} \frac{(-1)^{k-1}}{k!(d-1-k)!} \sum_{n=1}^r \frac{\Gamma\left(\frac{u}{\hbar}+n+1-r(d-k)-\frac{1}{2}\right)}{\Gamma\left(\frac{v}{\hbar}+r k+n-1+\frac{1}{2}\right)}\Bigg)\\
&=\Psi^*(u) \Psi(v).
\end{align*}
\subsubsection{Monotone Hurwitz numbers}
The monotone Hurwitz numbers were introduced in \cite{GGN} for studying the asymptotics of HCIZ matrix integral,
and are variations of the ordinary Hurwitz numbers.
The fermionic representation and quantum curve of this kind of Hurwitz numbers were studied in \cite{ALS, DK} and references therein.
The fermionic one-point function of this model and its dual are given by (see, for example, \cite{DK})
\begin{align*}
	\begin{aligned}
		\Psi(z)=&1+\sum_{k=1}^{\infty} \frac{\prod_{j=1}^{k-1}\frac{1}{1-j \hbar}}{k! \hbar^k} z^{- k} ,\\
		\Psi^*(z)=&1+\sum_{k=1}^{\infty} \frac{(-1)^k\prod_{j=1}^{k-1}\frac{1}{1+j \hbar}}{k! \hbar^k} z^{- k}.\\
	\end{aligned}
\end{align*}
The lifting operator and its adjoint operator can be obtained as
$$
l_z=\hat{x}(1+\hat{x}\hat{y}+\hbar)=(1+\hbar z\partial_z)z,
\qquad l^*_z=(1-\hbar z\partial_z)z.
$$
Indeed, one can check the following equation for the basis $\Phi^{mm}_k$ corresponding to a point in Sato-Grassmannian defined in \cite{ALS}
$$\hat{x}(1+\hat{x}\hat{p}+\hbar)\Phi^{mm}_k=(1+k\hbar)\Phi^{mm}_{k+1}$$
after specializing to the monotone Hurwitz number case. This way, one can apply Theorem \ref{thm:main KP} to compute the fermionic two-point function
$$
\begin{aligned}
\Psi(u,v)=&\frac{1}{u-v}+\sum_{m,n\ge 0}\frac{(-1)^nu^{-m-1}v^{-n-1}}{\hbar^{m+n+1}(m+n+1)m!n!}
\cdot \prod_{j=-m}^{n}\frac{1}{1+j\hbar}
\end{aligned}
$$
by solving the equation \eqref{keyequation} as:
\begin{align*}
(l_u^*-l_v)\Psi(u,v)=&\frac{u(u-v(1-\hbar))}{(u-v)^2}-\frac{v(u(1+\hbar)-v)}{(u-v)^2}\\
&+\sum_{m,n\ge 0}\frac{(-1)^n(1+m\hbar)\prod_{j=-m}^{n}\frac{1}{1+j\hbar}}{\hbar^{m+n+1}(m+n+1)m!n!}u^{-m}v^{-n-1}
\\
&-\sum_{m,n\ge 0}\frac{(-1)^n(1-n\hbar)\prod_{j=-m}^{n}\frac{1}{1+j\hbar}}{\hbar^{m+n+1}(m+n+1)m!n!}u^{-m-1}v^{-n}
\\
=&\Psi^*(u) \Psi(v).
\end{align*}

\subsubsection{Grothendieck's dessins d'enfants}
The Grothendieck's dessins d'enfants are certain graphs embedded in Riemann surfaces and was introduced by Grothendieck \cite{Gr}.
This combinatorial object is related to the Belyi pair and plays an important role in studying field of algebraic numbers $\bar{\mathbb{Q}}$.
Later,
the integrability of partition function of counting dessins was studied in \cite{KZ, Z19a}.
We recommend \cite{KZ} for details about this model.
The fermionic one-point function corresponding to this model is
$$
\Psi(z) =\sum\limits_{k=0}^\infty \frac{\prod_{j=0}^{k-1}(\alpha+j \hbar)(\beta+j\hbar)}{k! \hbar ^k} z^{-k}
$$
for parameters $\alpha,\beta$ and one can take the lifting operator and its adjoint operator as
\footnote{ One can check  the following equation   for the basis $\phi_k$ defined in  \cite{Z19a} :
$$
(z  + \hbar z \partial_z -k\cdot \hbar) \phi_k =   \phi_{k+1}  +\tfrac{ \textstyle \prod_{i=0}^k (\alpha-i\hbar)(\beta-i\hbar)\phi_{0}}{(k+1)!}
$$
and hence the operator  $z  + \hbar z \partial_z $ is a lifting operator.
Here our  $\hbar,\alpha,\beta$ equals to the $s,s u,sv$    in  \cite{Z19a} respectively. }
$$
l_z=x+ x \hat p = z  + \hbar z \partial_z,\qquad l^*_z=z  -\hbar  \partial_zz.
$$
By using the Theorem \ref{thm:main KP}, we can compute the fermionic two-point function, which is given by
\begin{align*}
	\begin{split}
	\Psi(u,v)=&\frac{1}{u-v}+\sum_{m,n\geq 0}\frac{(-1)^n\prod_{j=-n}^m(\alpha+j \hbar )(\beta+j \hbar )}
	{(m+n+1) m!n! \hbar^{m+n+1}} u^{-n-1} v^{-m-1},
	\end{split}
\end{align*}
and matches with the coordinates of corresponding model in the affine Grassmannian computed in \cite{Z19a}. In fact, one can check the following formula from equation \eqref{keyequation}:
\begin{align*}
(l_u^*-l_v)&\Psi(u,v)=\frac{u(u-v)+v\hbar}{(u-v)^2}-\frac{v(u-v+\hbar)}{(u-v)^2}\\
&+\sum_{m,n\geq 0}\frac{(-1)^n(u+\hbar n)\prod_{j=-n}^m(\alpha+j \hbar )(\beta+j \hbar )}
	{(m+n+1) m!n! \hbar^{m+n+1}} u^{-n-1} v^{-m-1}\\
&-\sum_{m,n\geq 0}\frac{(-1)^n(v-\hbar m-\hbar)\prod_{j=-n}^m(\alpha+j \hbar )(\beta+j \hbar )}
	{(m+n+1) m!n! \hbar^{m+n+1}} u^{-n-1} v^{-m-1}\\
=&\Psi^*(u) \Psi(v).
\end{align*}

\subsubsection{Spin Hurwitz numbers with completed cycles}
The spin Hurwitz numbers with completed $(r+1)$-cycles were introduced and studied in \cite{EOP, GKL}.
It was computed by Wang and the second- and third-named authors that their BKP-affine coordinates are given by (see equation (66) in \cite{JWY}),
\begin{equation*}
	\begin{split}
		&a_{0,n}^{(r,\vartheta)} = -a_{n,0}^{(r,\vartheta)} =
		\frac{\hbar^{-n}}{ 2\cdot n!}\cdot \exp\big( \hbar^r\frac{n^{r+1}}{r+1} \big),
		\qquad \forall n>0,\\
		&a_{n,m}^{(r,\vartheta)} = \frac{\hbar^{-m-n}}{4\cdot m!\cdot n!} \cdot \frac{m-n}{m+n} \cdot
		\exp \big( \hbar^r\frac{m^{r+1} + n^{r+1}}{r+1} \big),
		\qquad \forall m,n>0,
	\end{split}
\end{equation*}
where we have replaced the notations $p$ and $\beta$ in \cite{JWY} by $\hbar^{-1}$ and $\hbar^r$, respectively, corresponding to the Riemann--Hurwitz formula.
Thus,
as a tau-function of the BKP hierarchy,
the generating function of spin Hurwitz numbers with completed cycles has the first and second fermionic one-point functions as
\begin{align*}
	\Psi^B(z)=&
	\frac{1}{2}
	+\sum_{n>0}(-1)^{n}
	\frac{\hbar^{-n}}{2\cdot n!} \cdot
	\exp \big( \hbar^r \frac{n^{r+1}}{r+1} \big)
	\cdot z^{-n},\\
	\tilde{\Psi}^B(z)=&
	-z+\frac{\hbar^{-1}}{2}\exp \big( \frac{\hbar^r}{r+1} \big)\\
	&\qquad-\sum_{n>0}(-1)^{n}
	\frac{\hbar^{-n-1}(n-1)}{2\cdot (n+1)!} \cdot
	\exp \big( \hbar^r\frac{n^{r+1} + 1}{r+1} \big)
	\cdot z^{-n},
\end{align*}
and the fermionic two-point function is given as
\begin{align*}
\Psi^B(u,v)=&\frac{u-v}{2(u+v)}
+\sum_{n>0}(-1)^{n+1}\frac{e^{\frac{\hbar^r n^{r+1}}{r+1}}}{2\hbar^n n!}(u^{-n}-v^{-n})\\
&+\sum_{m,n>0}\frac{(-1)^{m+n+1}(m-n)\hbar^{-m-n}e^{h^r\frac{m^{r+1}+n^{r+1}}{r+1}}}{2m!n!(m+n)}u^{-n}v^{-m}.
\end{align*}
The lifting operator for this model could be (see Theorem 5.1 in \cite{JWY})
\begin{align*}
	l_z
	=\exp\Big({(-1)^r\frac{\hbar^r}{r+1}}\Big)
	\cdot\exp\Big(- \frac{\hbar^r}{r+1}\cdot \sum_{i=0}^{r}z
	\big(z\frac{\partial}{\partial z}\big)^iz^{-1}\big(z\frac{\partial}{\partial z}\big)^{r-i}\Big)z.
\end{align*}
Since $\iota(l_z)=l_z$,
the anti-symmetrization of this lifting operator is itself,
i.e., $\tilde{l}_z=l_z$.
Using the above data,
one can directly check
\begin{align*}
	(l_u&+l_v)\cdot\Psi^B(u,v)=\sum_{n\ge0}\frac{(-1)^ne^{n\frac{\hbar^rn^r}{r+1}}}{2\cdot n!\hbar^n}(uv^{-n}-vu^{-n})\\
&+\sum_{m,n>0}\frac{(-1)^{m+n+1}e^{\frac{\hbar^r (n^{r+1}+m^{r+1})}{r+1}}}{4\hbar^{m+n+1}n!m!}\left(\frac{e^{\frac{\hbar^r n^{r}(m-1)}{r+1}}}{m+1}-\frac{e^{\frac{\hbar^r m^{r}(n-1)}{r+1}}}{n+1}\right)u^{-n}v^{-m}\\
=&\Psi^B(u)\tilde{\Psi}^B(v)-\Psi^B(v)\tilde{\Psi}^B(u),
\end{align*}
which matches with the equation \eqref{eqn:BKPmain} in Theorem \ref{thm:B main}.

	\renewcommand{\refname}{Reference}
	\bibliographystyle{plain}

\begin{thebibliography}{399}

		\bibitem[ACM]{ACM} M. Adler, M. Cafasso, P. van Moerbeke,
		{\it Non-lineal PDEs for gap probabilities in random matrices and KP theory},
		Physica D, 241 (2012), 2265-2284.
		
		\bibitem[AM]{AM} M. Adler, P. van Moerbeke,
		{\it A matrix integral solution to two-dimensional W(p) gravity},
		Commun. Math. Phys. 147, 25 (1992)
		
		\bibitem[ADKMV]{ADKMV} M. Aganagic, R. Dijkgraaf, A. Klemm, M. Mari\~{n}o, C. Vafa.
		{\it Topological strings and integrable hierarchies.}
		Commun. Math. Phys. 261, 451–516 (2006).
		
		\bibitem[A1]{A15} A. Alexandrov,
		{\it Enumerative geometry, tau-functions and Heisenberg--Virasoro algebra},
		Commun. Math. Phys, 338(1):195-249,2015.
		
		\bibitem[A2]{A21} A.Alexandrov,
		{\it KP integrability of triple Hodge integrals. II. Generalized Kontsevich matrix model},
		Anal.Math.Phys. 11, 24 (2021).
		
		\bibitem[A3]{A23} A. Alexandrov,
		{\it Generalized Br\'{e}zin--Gross--Witten tau function as a hypergeometric solution of the BKP hierarchy},
		Adv. Math. Vol 412 (2023), 108809.
		

        \bibitem[ALS]{ALS} A. Alexandrov, D. Lewanski, and S. Shadrin,
        {\it Ramifications of Hurwitz theory, KP integrability and quantum curves},
        Journal of High Energy Physics, 2016(5), 1-31.
		
		\bibitem[BY]{BY} F. Balogh, D. Yang,
		{\it Geometric interpretation of Zhou's explicit formula for the Witten--Kontsevich tau function},
		Lett. in Math. Phy., 2017, 107(10):1837-1857.
		
		\bibitem[BG]{BG} E. Br\'{e}zin and D. Gross,
		{\it The external field problem in the large N limit of QCD},
		Phys. Lett. B 97 (1980), no. 1, p120-124.
		
		\bibitem[BH]{BH} F. Balogh, J. Harnad,
		{\it Tau functions and their applications},
		Cambridge Monographs on Mathematical Physics. Cambridge University Press, 2021.
		
		
		\bibitem[CG]{CG} C. Chen, S. Guo,
		{\it Quantum curve and bilinear Fermionic form for the orbifold Gromov--Witten theory of $\mathbb{P}[r]$},
		Acta Mathematica Sinica, English Series, 2024, Vol. 40, No. 1, p43–80.
		
		\bibitem[CGG]{CGG} N. K. Chidambaram and E. Garcia-Failde and A. Giacchetto,
		{\it Relations on {$\overline{\mathcal{M}}_{g,n}$} and the negative $r$-spin Witten conjecture},
		arXiv:2205.15621

            \bibitem[DK]{DK} N. Do and M. Karev,
            {\it Monotone orbifold Hurwitz numbers}, Journal of Mathematical Sciences, 226, 568-587. (2017).
		
		\bibitem[DJM]{DJM} E. Date, M. Jimbo, T. Miwa,
		{\it Solitons: differential equations, symmetries and infinite dimensional algebras},
		Cambridge University Press, (2000).
		
		\bibitem[DLM]{DLM} X.-M. Ding, Y. Li and L. Meng,
		{\it From $r$-spin intersection numbers to Hodge integrals},
		J. High Energy Phys., 2016, 15.

            \bibitem[DMNPS]{DMNPS} P. Dunin-Barkowski, M. Mulase, P. Norbury, A. Popolitov and S. Shadrin, {\it Quantum spectral curve for the Gromov--Witten theory of the complex projective line}, 	
            J. Reine Angew. Math. 726 (2017), 267-289.
		
		\bibitem[DYZ]{DYZ} B. Dubrovin, D. Yang and D. Zagier,
		{\it On tau-functions for the KdV hierarchy},
		Selecta Math. (N.S.) 27 (2021), no. 12, 47pp.

        \bibitem[DZ]{DZ} F. Deng and J. Zhou,
	{\it On fermionic representation of the framed topological vertex}, J. High Energ. Phys. 2015, 1–22 (2015).
		
		\bibitem[EOP]{EOP} A. Eskin, A. Okounkov and R. Pandharipande,
		{\it The theta characteristic of a branched covering}, Adv. Math., 217, 3 (2008): 873-888.
		
		
    \bibitem[FSZ]{FSZ} C. Faber, S. Shadrin, D. Zvonkine,
		{\it Tautological relations and the r-spin Witten conjecture},  Ann. Sci. \'{E}c. Norm. Sup\'{e}r. (4) 43 (2010), no. 4, 621-658.
		
		\bibitem[FKN]{FKN} M. Fukuma, H. Kawai, R. Nakayama,
		{\it Infinite dimensional Grassmannian structure of two-dimensional quantum gravity},
		Comm. in Math. Phys. 143 (1992), p371-403.
		
		\bibitem[GKL]{GKL} A. Giacchetto, R. Kramer, and D. Lewa\'nski,
		{\it A new spin on Hurwitz theory and ELSV via theta characteristics}, arXiv:2104.05697.
		
		\bibitem[GGN]{GGN}I. P. Goulden, M. Guay-Paquet and J. Novak,
		{\it Monotone {H}urwitz numbers and the {HCIZ} integral},
		Ann. Math. Blaise Pasca, 21 (2014), no. 1, 71-89.
		
		\bibitem[GN]{GN}D. Gross and M. Newman,
		{\it Unitary and Hermitian matrices in an external field. II. The Kontsevich model and continuum Virasoro constraints},
		Nuclear Phys. B 380 (1992), no. 1-2, 168-180.
		
		\bibitem[GW]{GW} D. Gross and E. Witten,
		{\it Possible third order phase transition in the large N lattice gauge theory}, Phys. Rev. D 21.446 (1980).
		
		\bibitem[Gr]{Gr} A. Grothendieck,
		{\it Esquisse d'un programme},
		London Mathematical Society Lecture Note Series, 1997: 5-48.
		
		\bibitem[GJYZ]{GJYZ} S. Guo, C. Ji, C. Yang and Q. Zhang,
		{\it } In preparation.
		
		\bibitem[JWY]{JWY} C. Ji, Z. Wang and C. Yang,
		{\it Kac--Schwarz operators of type $B$, quantum spectral curves, and spin Hurwitz numbers}, Journal of Geometry and Physics,
		189 (2023): 104831.
		
		\bibitem[KMMMZ]{KMMMZ} S. Kharchev, A. Marshakov, A. Mironov, A. Morozov and A. Zabrodin,
		{\it Towards unified theory of 2-d gravity},
		Nucl. Phys. B 380 (1992) 181.
		
		\bibitem[KS]{KS}V. Kac and A. Schwarz,
		{\it Geometric interpretation of the partition function of 2D gravity},
		Physics Letters B, 1991, 257(3-4):329-334.

        \bibitem[KZ]{KZ} M. Kazarian and P. Zograf,
        {\it Virasoro constraints and topological recursion for Grothendieck's dessin counting.}
        Lett. Math. Phys. (2015) 105:1057–1084.

        \bibitem[KL]{KL} M. Kazarian and S. Lando,
        {\it An algebro-geometric proof of Witten’s conjecture}, J. Amer. Math. Soc. 20 (2007), no. 4, 1079–1089.

		
		\bibitem[K]{K} M. Kontsevich,
		{\it Intersection theory on the moduli space of curves and the matrix Airy function,}
		Commun. Math. Phys. 147 (1992), 1-23.


		
		\bibitem[L]{L} J. Leur,
		{\it The Adler-Shiota-van Moerbeke formula for the BKP hierarchy},
		J. Math. Phys., 36 (9):4940-4951, 1995.
		
		\bibitem[LY1]{LY2} X. Liu and C. Yang,
		{\it Q-polynomial expansion for Br\'ezin--Gross--Witten tau-function,}
		Adv. Math., Vol 404 (2022) 108456.
		
		\bibitem[LY2]{LY3} X. Liu and C. Yang,
		{\it Action of $W-$type operators on Schur functions and Schur Q-functions},
		to appear in J. London Math. Soc..
		
		\bibitem[LY3]{LY1} X. Liu and C. Yang,
		{\it Schur Q-polynomials and Kontsevich--Witten tau function},
		Peking Math. J., Vol 7 (2024), 713–758.
		
		\bibitem[MM]{MM} A. Mironov and A. Morozov,
		{\it Superintegrability of Kontsevich matrix model},
		Eur. Phys. J. C., 81, 270 (2021).
		
		\bibitem[N]{N17} P. Norbury,
		{\it A new cohomology class on the moduli space of curves}, Geom. Topol. 27 (2023) 7, 2695-2761.
		
		\bibitem[O1]{O00}A. Okounkov,
            {\it Toda equations for Hurwitz numbers},
            Math. Res. Lett. 7, no. 4, 447–453 (2000).
		

        \bibitem[O2]{O01}A. Okounkov,
		{\it Random matrices and random permutations},
		Int. Math. Res. Not., 20, 2000, 1043-1095.
		
		\bibitem[O3]{O02}A. Okounkov,
		{\it Generating functions for intersection numbers on moduli spaces of curves},
		Int. Math. Res. Not. 2002, no. 18, 933-957.
		
		\bibitem[S]{S} M. Sato,
		{\it Soliton equations as dynamical systems on an infinite dimensional Grassmannian manifold},
		RIMS Kokyuroku, 1981, 439: 30-46.
		
		\bibitem[TW]{TW} C. A. Tracy, H. Widom,
		{\it Level-spacing distributions and the Airy kernel},
		Commun. Math. Phys., 159, 1994, 151-174.
		
		\bibitem[WY]{WY} Z. Wang, C. Yang,
		{\it BKP hierarchy and connected bosonic $N$-point functions},
		Lett. in Math. Phy., 2022, 112:62.
		
		\bibitem[WYZ]{WYZ} Z. Wang, C. Yang, Q. Zhang,
		{\it BKP-affine coordinates and emergent geometry of generalized Brézin--Gross--Witten tau-functions},
		arXiv: 2301.01131.
		
		\bibitem[W1]{W90} E. Witten,
		{\it Two-dimensional gravity and intersection theory on moduli space}.
		Surveys in differential geometry (Cambridge, MA, 1990), 243-310, Lehigh Univ., Bethlehem, PA, 1991.
		
		\bibitem[W2]{W} E. Witten,
		{\it Algebraic geometry associated with matrix models of two-dimensional Gravity}.
		Topological methods in modern mathematics (Stony Brook, NY
		1991), 235-269, Publish or Perish, Houston, TX, 1993.

		
		\bibitem[Z1]{Z13} J. Zhou,
		{\it Explicit formula for Witten--Kontsevich tau-function.}
		arXiv:1306.5429.
		
		\bibitem[Z2]{Z15} J. Zhou,
		{\it Emergent geometry and mirror symmetry of a point.}
        arXiv:1507.01679.

        \bibitem[Z3]{Z19a} J. Zhou,
        {\it Grothendieck's dessins d'enfants in a web of dualities.}
        arXiv:1905.10773.

	\end{thebibliography}

	\vspace{30pt} \noindent
\end{document}